\documentclass[11pt]{article}

\usepackage{amsmath}
\usepackage{amssymb}
\usepackage{amsthm}
\usepackage{amsfonts}
\usepackage{graphicx}
\usepackage{textcomp}
\usepackage{color}
\usepackage{hyperref}
\usepackage{multicol}
\usepackage{slashed}
\usepackage{tikz-cd}
\usepackage[utf8]{inputenc}
\usepackage[english]{babel}
\usepackage[a4paper, total={5.7in, 9.5in}]{geometry}

\numberwithin{equation}{section}

\newtheorem{Proposition}{Proposition}[section]

\newtheorem{Theorem}{Theorem}[section]
\newtheorem{Corollary}{Corollary}[section]
\newtheorem{Definition}{Definition}[section]

\begin{document}

\begin{titlepage}

\title{Metric-Connection Geometries on Pre-Leibniz Algebroids: A Search for Geometrical Structure in String Models}

\author{Tekin Dereli\footnote{tdereli[at]ku.edu.tr}, \quad Keremcan Do\u{g}an\footnote{kedogan[at]ku.edu.tr} \\ \small Department of Physics, Ko\c{c} University, 34450 Sar{\i}yer, \.{I}stanbul, Turkey}

\date{}

\maketitle

\begin{abstract} 
\noindent The metric-affine and generalized geometries, respectively, are arguably the appropriate mathematical frameworks for Einstein's theory of gravity and the low-energy effective massless oriented closed bosonic string field theory. In fact, mathematical structures in a metric-affine geometry are written on the tangent bundle, which is itself a Lie algebroid; whereas those in generalized geometries introduced as the basis of double field theories, are written on Courant algebroids. The Lie, Courant and the higher Courant algebroids used in exceptional field theories, are all special cases of pre-Leibniz algebroids. Provided with some additional ingredients, the construction of such geometries can all be carried over to regular pre-Leibniz algebroids. We define below the notions of locality structures and locality projectors, which are some such necessary ingredients. In terms of these structures, $E$-metric-connection geometries are constructed with (possibly) a minimum number of assumptions. Certain small gaps in the literature are also filled as we go along. $E$-Koszul connections, as a generalization of Levi-Civita connections, are going to be defined and shown to be helpful for some results including a simple generalization of the fundamental theorem of Riemannian geometry. We also show that metric-affine geometries can be constructed in a unique way as special cases of $E$-metric-connection geometries. Moreover, generalized geometries are shown to follow as special cases, and various properties of linear generalized-connections are proven in the present framework. Similarly, uniqueness of the locality projector in the case of exact Courant algebroids is proven; a result that explains why the curvature operator, defined with a projector in the double field theory literature is a necessity.
\end{abstract}

\vskip 2cm


\thispagestyle{empty}

\end{titlepage}

\maketitle

\section{Introduction}
\noindent Geometric structures play an important role in classical field theories. For instance, Einstein's field equations of gravity on a smooth manifold $M$ can be derived from the variations of the Einstein-Hilbert action 
\begin{equation} S_{EH}[g] = \int_M{R(^g \nabla, g)} *_g 1,
\label{ea1}
\end{equation}
where the Lagrangian density is given by a geometric quantity called the Ricci scalar $R(^g \nabla, g)$ corresponding to a metric $g$ and its associated Levi-Civita connection $^g \nabla$, which is metric-$g$-compatible and torsion-free. If an arbitrary affine connection rather than the Levi-Civita connection is considered, then mild generalizations of general relativity can be deduced. For example, Einstein-Cartan gravity can be written in terms of a metric and a metric-compatible but torsionful affine connection \cite{1}. The most general theory would be the case where the torsionful affine connection is not metric-compatible. All of these structures are defined on the tangent bundle, which is a Lie algebroid, and its dual cotangent bundle. By constructing analogous objects, metric-affine geometries can be written on arbitrary Lie algebroids \cite{2}, \cite{3}.

Another example of theories of ``geometric'' origin comes from the low energy effective massless oriented closed bosonic string theory equations, which can be derived from the variations of the following action
\begin{equation}
S[g, H, \phi] = \int_M{e^{-2 \phi} \left[ R(^g \nabla, g) *_g 1 - \frac{1}{2} H \wedge *_g H + 4 d \phi \wedge *_g d \phi \right]},
\label{ea2}
\end{equation}
where $\phi$ is a smooth function called the dilaton, and $H$ is a closed 3-from, which is the field strength of the Kalb-Ramond field. This action can be written in an Einstein-Hilbert-like fashion by using local double field theory \cite{4}. It is a classical field theory on a so-called ``doubled-manifold'' in order to incorporate $T$-duality of string theory as a true symmetry. The constructions of this theory are closely related to Hitchin's generalized geometry \cite{5}\footnote{In local double field theory, almost identical structures are used, but they are constrained by section conditions. In the double field theory literature, these structures are often mathematically ``ill-defined''.}. In this setting, one can rewrite the above action (\ref{ea2}) in terms of a generalized geometric quantity called the generalized-Ricci scalar corresponding to a generalized-metric and a generalized-Levi-Civita connection \cite{6}\footnote{This action can be written in the usual geometric setting if one considers Riemann-Cartan-Weyl geometry \cite{7}.}. In order to define these analogous structures, one needs to use the language of exact Courant algebroids and ``generalize'' the quantities on the tangent bundle. Moreover, exceptional field theories, which include also $U$-duality, can be constructed on ``higher'' exact Courant algebroids \cite{8}.

Lie algebroids, Courant algebroids and higher Courant algebroids are all special cases of pre-Leibniz algebroids. Hence, a natural question arises about the formulation of metric-connection geometries on an arbitrary pre-Leibniz algebroid. Regular pre-Leibniz algebroids that are endowed with a locality structure allow one to construct structures such as linear connections, metric, torsion, curvature and non-metricity tensors \cite{9}. This can be done in a way that restrictions on Lie, Courant and higher Courant algebroids yield the structures in the usual, generalized and exceptional generalized geometries respectively. Most of the works on local double field theory and generalized geometry literature focus on generalized-Levi-Civita connections with a small number of exceptions such as the teleparallel local double field theory \cite{10} and deformed Weitzenb\"{o}ck connections \cite{11}. The most arbitrary linear generalized-connections should be studied to have a better understanding of generalized geometries. Working on such linear connections in the pre-Leibniz algebroid setting is one of the main purposes of this paper.

The organization of the paper is as follows. In section 2, after introducing the notation of the paper, a detailed summary of metric-affine geometries will be given in order to see the analogies between geometries. Constructions such as metric, affine connection, corresponding torsion, curvature, non-metricity, Ricci tensors, and Ricci scalar will be defined. In section 3, pre-Leibniz algebroids will be introduced and ``generalized'' $E$-versions of the previous structures will be constructed for regular local pre-Leibniz algebroids endowed with locality projectors. In this section, we will fill in some small gaps in the literature by making the constructions with a possibly minimum number of assumptions and parallel to metric-affine case. We will define locality structures and locality projectors in order to define $E$-curvature. As analogous to Levi-Civita connections, we will define $E$-Koszul connections, which will satisfy some interesting properties and will be useful for some results. For example, we will prove a ``simple'' generalization of the fundamental theorem of Riemannian geometry. Moreover, we will show that $E$-metric-connection geometry yields a unique metric-affine geometry. In section 4, generalized geometries on pre-Courant algebroids will be constructed as a special case, and some general results which will hold for the case of exact Courant algebroids will be proven. Moreover, we will prove that for exact almost-Courant algebroids, there is a unique locality projector, which is the one already used in the double field theory literature. In the last section, concluding remarks and possible future research directions will be discussed.

%%%%%%%%%%%%%%%%%%%%%%%%%%%%%%%%%%%%%%%%%%%%%%%%%%%%%%

\section{Metric-Affine Geometries on a Smooth Manifold}

\noindent In this paper, every construction is assumed to be in the smooth category, and Einstein's summation convention for repeated indices is used. $M$ denotes an orientable (second countable, Hausdorff) manifold of dimension $n \in \mathbb{N}$. Its tangent bundle $T(M)$ is a real vector bundle of rank $n$ whose sections are vector fields. The set of vector fields is denoted by $\mathfrak{X}(M)$ and it forms a Lie algebra with the Lie bracket $[\cdot,\cdot]: \mathfrak{X}(M) \times \mathfrak{X}(M) \to \mathfrak{X}(M)$, which is anti-symmetric and satisfies the Jacobi identity and the Leibniz rule. $\mathfrak{X}(M)$ is isomorphic to derivations on the set of smooth functions $C^{\infty}(M, \mathbb{R})$, and the action of a vector field $V$ on a smooth function $f$ will be denoted by $V(f)$. The term, local frame, will be used only for a local basis of $T(M)$. On a local frame $(X_a)$, the Lie bracket satisfies $[X_a, X_b] = \gamma^c_{\ a b} X_c$ for some $\{ \gamma^c_{\ a b} \}$, which are called the anholonomy coefficients. A local frame is called holonomic if all anholonomy coefficients vanish\footnote{Such a local frame always exists as it is induced by the coordinate maps on a local trivialization chart.}. The tangent bundle's dual is the cotangent bundle $T^*(M)$ whose sections are exterior differential 1-forms. A local frame $(X_a)$ has the dual $\left( e^a \ | \ e^a(X_b) = \delta^a_{\ b} \right)$ called a local coframe, which is a local basis for $T^*(M)$, where $\delta^a_{\ b}$ is the Kronecker delta symbol. $(q, r)$-type tensors over $M$ are defined as the elements of
\begin{equation} Tens^{(q, r)}(M) := \Gamma \left( \bigotimes_{i = 1}^q T(M) \otimes \bigotimes_{j = 1}^r T^*(M) \right),
\label{eb1}
\end{equation}
where the set of sections of any fiber bundle $E$ over $M$ is denoted by $\Gamma(E)$. On a local frame $(X_a)$ over a coordinate chart $U$, the components of a $(q, r)$-type tensor $Z$ are defined by
\begin{equation} Z^{a_1 \ldots a_q}_{\ \ \ \ \ \ \ b_1 \ldots b_r} := Z(e^{a_1}, \ldots, e^{a_q}, X_{b_1}, \ldots, X_{b_r}),
\label{eb2}
\end{equation}
which are smooth functions over $U$. $\Omega^p(M)$ denotes the set of exterior differential $p$-forms, which are anti-symmetric $(0, p)$-type tensors. On the set of all exterior differential forms, there is a degree-$1$ graded derivation $d: \Omega^p(M) \to \Omega^{p + 1}(M)$ called the exterior derivative\footnote{$d^2 = 0$, so it defines the de Rham cohomology $H_{dR}(M)$.}, and a degree-$(-1)$ graded anti-derivation $\iota_V: \Omega^p(M) \to \Omega^{p - 1}(M)$ called the interior product with respect to a vector field $V$. The action of the Lie bracket $[V, .]$ can be extended into whole tensor algebra by the Lie derivative $\mathcal{L}_V: Tens^{(q, r)}(M) \to Tens^{(q, r)}(M)$ with respect to the vector field $V$. 

Metric-affine geometries are defined to be a triplet $(M, g, \nabla)$\footnote{As we will see, a more appropriate notation for the next sections would be the quadruplet $(M, (T(M), id_{T(M)}, [\cdot,\cdot], [0], 0), g, \nabla)$, where $id_X$ stands for the identity map in a set $X$.} where $M$ is a (smooth, orientable) manifold, $g$ is a symmetric and non-degenerate $(0, 2)$-type tensor called the metric\footnote{The metric $g$ induces an isomorphism, which will be denoted by the same symbol $g: \mathfrak{X}(M) \to \Omega^1(M)$ given by $g(u)(v) := g(u, v)$, for all $u, v \in T(M)$. It also induces another isomorphism $*_g: \Omega^p(M) \to \Omega^{n - p}(M)$, which is called the Hodge star isomorphism.}  and $\nabla$ is an affine connection on $M$\footnote{This is a linear vector bundle connection on $T(M)$.} defined to be an $\mathbb{R}$-bilinear map $\nabla: \mathfrak{X}(M) \times \mathfrak{X}(M) \to \mathfrak{X}(M)$, $(U, V) \mapsto \nabla_U V$ satisfying
\begin{align} \nabla_U(f V) &= U(f) V + f \nabla_U V, \nonumber\\
\nabla_{f U} V &= f \nabla_U V,
\label{eb3}
\end{align}
for all $U, V \in \mathfrak{X}(M), f \in C^{\infty}(M, \mathbb{R})$. On a local frame $(X_a)$, the connection coefficients $\{ \Gamma(\nabla)^c_{\ a b} \}$ are defined by
\begin{equation} \Gamma(\nabla)^a_{\ b c} := \langle e^a, \nabla_{X_b} X_c \rangle,
\label{eb4}
\end{equation}
where the map $\langle \cdot,\cdot \rangle : \Omega^p(M) \times \mathfrak{X}(M) \to \Omega^{p - 1}(M)$ is given by
\begin{equation} \langle \omega, V \rangle  := \iota_V \omega ,
\label{eb5}
\end{equation}
for all $\omega \in \Omega^p(M), V \in \mathfrak{X}(M)$\footnote{Note that $\Omega^0(M) = C^{\infty}(M, \mathbb{R})$.}. The action of an affine connection can be extended into the whole tensor algebra by the Leibniz rule, and it induces a map called the exterior covariant derivative, which will be denoted by the same symbol $\nabla: Tens^{(q, r)}(M) \to Tens^{(q, r + 1)}(M)$, $Z \mapsto \nabla Z$
\begin{align} (\nabla Z)(\omega_1, \ldots, \omega_q, U, V_1, \ldots V_r) &:= (\nabla_U Z)(\omega_1, \ldots, \omega_q, V_1, \ldots V_r) \nonumber\\
&:= U \left( Z(\omega_1, \ldots, \omega_q, V_1, \ldots, V_r) \right) \nonumber\\
& \qquad - \sum_{i = 1}^q Z(\omega_1, \ldots, \nabla_U \omega_i, \dots, \omega_q, V_1, \ldots, V_r) \nonumber\\ 
& \qquad - \sum_{j = 1}^r Z(\omega_1, \dots, \omega_q, V_1, \ldots, \nabla_U V_j, \ldots, V_r),
\label{eb6}
\end{align}
for all $\omega_i \in \Omega^1(M), V_j, U \in \mathfrak{X}(M)$.

The non-metricity tensor corresponding to an affine connection $\nabla$ and a metric $g$ is defined to be the following $(0, 3)$-type tensor
\begin{equation} Q(\nabla, g) := \nabla g.
\label{eb7}
\end{equation}
If $Q(\nabla, g) = 0$, then the affine connection $\nabla$ is called metric-$g$-compatible. Every manifold with a metric $g$ admits a metric-$g$-compatible connection. On a local frame $(X_a)$, the non-metricity components read
\begin{equation} Q(\nabla, g)_{a b c} = X_a(g_{bc}) - \Gamma(\nabla)^d_{\ a b} g_{d c} - \Gamma(\nabla)^d_{\ a c} g_{b d}.
\label{eb8}
\end{equation}

The torsion operator of an affine connection $\nabla$ is defined as a map $T(\nabla): \mathfrak{X}(M) \times \mathfrak{X}(M) \to \mathfrak{X}(M)$,
\begin{equation} T(\nabla)(U, V) := \nabla_U V - \nabla_V U - [U, V],
\label{eb9}
\end{equation}
for all $U, V \in \mathfrak{X}(M)$. If $T(\nabla)(U, V)$ vanishes for all $U, V \in \mathfrak{X}(M)$, then $\nabla$ is called torsion-free. $T(U, V)$ is $C^{\infty}(M, \mathbb{R})$-bilinear, so the torsion tensor, which will be denoted by the same symbol, is a $(1, 2)$-type tensor defined as
\begin{equation} T(\nabla)(\omega, U, V) := \langle \omega, T(\nabla)(U, V) \rangle ,
\label{eb10}
\end{equation}
for all $\omega \in \Omega^1(M), U, V \in \mathfrak{X}(M)$. On a local frame $(X_a)$, the torsion components read
\begin{equation} T(\nabla)^a_{\ b c} = \Gamma(\nabla)^a_{\ b c} - \Gamma(\nabla)^a_{\ c b} - \gamma^a_{\ b c}.
\label{eb11}
\end{equation}
Note that due to anti-symmetry of the Lie bracket, the torsion operator is also anti-symmetric so that $T(\nabla)^a_{\ b c}$ is anti-symmetric in $b$ and $c$.

The curvature operator of an affine connection $\nabla$ is defined as a map $R(\nabla): \mathfrak{X}(M) \times \mathfrak{X}(M) \times \mathfrak{X}(M) \to \mathfrak{X}(M)$, 
\begin{equation} R(\nabla)(U, V, W) := \nabla_U \nabla_V W - \nabla_V \nabla_U W - \nabla_{[U, V]}W,
\label{eb12}
\end{equation}
for all $U, V, W \in \mathfrak{X}(M)$. $R(\nabla)(U, V, W)$ is $C^{\infty}(M, \mathbb{R})$-multilinear, so the curvature tensor is defined as a $(1, 3)$-type tensor given by
\begin{equation} R(\nabla)( \omega, U, V, W) := \langle \omega, R(\nabla)(U, V, W) \rangle ,
\label{eb13}
\end{equation}
for all $\omega \in \Omega^1(M), U, V, W \in \mathfrak{X}(M)$. On a local frame $(X_a)$, the curvature components read
\begin{align} R(\nabla)^a_{\ b c d} = & \ X_b \left( \Gamma(\nabla)^a_{\ c d} \right) - X_c \left( \Gamma(\nabla)^a_{\ b d} \right) + \Gamma(\nabla)^e_{\ c d} \Gamma(\nabla)^a_{\ b e} \nonumber\\
& - \Gamma(\nabla)^e_{\ b d} \Gamma(\nabla)^a_{\ c e} - \gamma^e_{\ b c} \Gamma(\nabla)^a_{\ e d}.
\label{eb14}
\end{align}
Ricci tensor $Ric(\nabla)$ of an affine connection $\nabla$ is defined as the trace of the linear map $U \mapsto R(\nabla)(U, V, W)$. On a local frame $(X_a)$, its components can be written as
\begin{equation} Ric(\nabla)_{a b} = R(\nabla)^c_{\ c a b}.
\label{eb15}
\end{equation}
The Ricci scalar or scalar curvature $R(\nabla, g)$ of an affine connection $\nabla$ and a metric $g$ is defined as the trace of $Ric(\nabla)$ with respect to $g$. On a local frame $(X_a)$, it can be written as
\begin{equation} R(\nabla, g) = Ric(\nabla)_{a b} g^{a b}.
\label{eb16}
\end{equation}
According to Vermeil's theorem, the Ricci scalar $R(\nabla, g)$ is the only scalar invariant that is linear in the second derivatives of the metric field $g$.

The fundamental theorem of Riemannian geometry states that there is a unique torsion-free, metric-$g$-compatible affine connection $^g \nabla$ called the Levi-Civita connection associated to $g$ given by the Koszul formula:
\begin{align} 2 g(^g \nabla_U V, W) = & \ U(g(V, W)) + V(g(U, W)) - W(g(U, V)) \nonumber\\
& - g([V, W], U) - g([U, W], V) + g([U, V], W),
\label{eb17}
\end{align}
for all $U, V, W \in \mathfrak{X}(M)$ whose components read
\begin{equation} \Gamma(^g \nabla)^a_{\ b c} = \frac{1}{2} g^{a d} \left[ X_b \left( g_{c d} \right) + X_c \left( g_{b d} \right) - X_d \left( g_{b c} \right) - \gamma^e_{\ c d} g_{e b} - \gamma^e_{\ b d} g_{e c}  + \gamma^e_{\ b c} g_{e d} \right]
\label{eb18}
\end{equation}
on a local frame $(X_a)$. Moreover, given a metric $g$, torsion $T(\nabla)$ and non-metricity $Q(\nabla, g)$ tensors, one can uniquely determine the affine connection $\nabla$ with the specified torsion and non-metricity tensors \cite{12}:
\begin{align} \Gamma(\nabla)^a_{\ b c} = & \Gamma(^g \nabla)^a_{\ b c} + \frac{1}{2} g^{a d} \Big[- Q(\nabla, g)_{b d c} + Q(\nabla, g)_{d c b} - Q(\nabla, g)_{c b d} \nonumber\\
& \qquad \qquad \qquad \quad \ - g_{e c} T(\nabla)^e_{\ b d} + g_{e d} T(\nabla)^e_{\ b c} - g_{e b} T(\nabla)^e_{\ c d} \Big].
\label{eb19} 
\end{align}

%%%%%%%%%%%%%%%%%%%%%%%%%%%%%%%%%%%%%%%%%%%%%%%%%%%%%%

\section{Metric-Connection Geometries on Local Pre-Leibniz Algebroids}

\noindent The aim of this section is to construct the geometric objects from the previous section in a more general setting and to fill some small gaps in the literature while proving some important results\footnote{One should be careful about the definitions of this section when comparing them with the ones that exist in literature. Assumptions of the defining properties of algebroids might change from paper to paper.}. Closely following \cite{9}, these constructions will be done with the minimum number of assumptions by copying the structures and properties of the tangent bundle, which is a vector bundle with the Lie bracket on its sections. For instance, tensors on a manifold can be easily generalized on arbitrary vector bundles because one can define dual vector bundles and tensor product of vector bundles. Let $E$ be a real vector bundle over a manifold $M$, with an abuse of notation, $(q, r)$-type $E$-tensors on $M$ are defined as 
\begin{equation}
Tens^{(q, r)}(E) := \Gamma \left( \bigotimes_{i = 1}^q E \otimes \bigotimes_{j = 1}^r E^* \right).
\label{ec1}
\end{equation}
Elements of $\mathfrak{X}(E) := \Gamma(E)$ are called $E$-vector fields. For any real vector bundle $E$, $\Gamma(E)$ is a module over $C^{\infty}(M, \mathbb{R})$, so one can construct a local basis for $\mathfrak{X}(E)$, and such a local basis $(X_a)$ is called a local $E$-frame. Its dual $\left( e^a \ | \ e^a(X_b) = \delta^a_{\ b} \right)$ is called a local $E$-coframe. Components of an $E$-tensor are defined in the usual way with respect to local $E$-frames and local $E$-coframes. Anti-symmetric $(0, p)$-type $E$-tensors are called $E$-exterior $p$-forms, and their set is denoted by $\Omega^p(E)$. With the anti-symmetrized tensor product, the set of all $E$-exterior $p$-forms becomes an anti-commutative graded algebra over $C^{\infty}(M, \mathbb{R})$. It has a degree-$(-1)$ graded anti-derivation $\iota_v : \Omega^p(E) \to \Omega^{p - 1}(E)$ defined by $\iota_v \Omega(u_1, \ldots, u_{p - 1}) := \Omega(v, u_1, \ldots, u_{p - 1})$ for all $\Omega \in \Omega^p(E), v, u_i \in \mathfrak{X}(E)$, which is called the $E$-interior product with respect to the $E$-vector field $v$. 

\begin{Definition} A $(0, 2)$-type $E$-tensor is called an $E$-metric if it is symmetric and non-degenerate\footnote{It is just a fiber-wise metric on $E$.}. 
\label{dc1}
\end{Definition}

\noindent Every $E$-metric $g$ induces an isomorphism, which will be denoted by the same symbol $g: \Gamma(E) \to \Gamma(E^*)$ given by $g(u)(v) := g(u, v)$ for all $u, v \in \Gamma(E)$. Let $\{ U_i \}$ be the connected components of $M$ and $rank(E|_{U_i}) = m_i \in \mathbb{N}$, then $g$ also induces an isomorphism $*_g: \Omega^p(E|_{U_i}) \to \Omega^{m_i - p}(E|_{U_i})$ for all $U_i$, which is called $E$-Hodge-star isomorphism, generalizing the usual Hodge-star \cite{13}. 

In order to define the $E$-version of affine connections, one needs to find a way to act on smooth functions via $E$-vector fields. If one considers anchored vector bundles, this can be done with the help of usual vector fields. An anchored vector bundle over $M$ is a doublet $(E, \rho)$ where $E$ is a vector bundle over $M$, and $\rho: E \to T(M)$ is a vector bundle morphism over $id_M$\footnote{Any vector bundle morphism $\psi: E \to F$ over $id_M$ induces a map, which will be denoted by the same letter, $\psi: \Gamma(E) \to \Gamma(F)$ defined by $\psi(u)(m) := \psi(u(m))$ for all $u \in \Gamma(E), m \in M$. The opposite is also true, so these two maps will be used interchangeably.}. 
\begin{Definition} For an anchored vector bundle $(E, \rho)$, a linear $E$-connection is defined as an $\mathbb{R}$-bilinear map $\nabla: \mathfrak{X}(E) \times \mathfrak{X}(E) \to \mathfrak{X}(E), (u, v) \mapsto \nabla_u v$ satisfying
\begin{align} \nabla_u (f v) &= \rho(u)(f) v + f \nabla_u v, \nonumber\\
\nabla_{f u} v &= f \nabla_u v,
\label{ec2}
\end{align}
for all $u, v \in \mathfrak{X}(E), f \in C^{\infty}(M, \mathbb{R})$ \cite{2}\footnote{As in the usual case, an $E$-connection can be defined on any vector bundle $F$ as a map $\nabla: \mathfrak{X}(E) \times \mathfrak{X}(F) \to \mathfrak{X}(F)$ satisfying the same properties. An $E$-connection on $E$ itself is called a linear $E$-connection.}. 
\label{dc2}
\end{Definition}

\noindent As in the usual case, on a local $E$-frame $(X_a)$, $E$-connection coefficients can be defined as
\begin{equation} \Gamma(\nabla)^a_{\ b c} := \langle e^a, \nabla_{X_b} X_c \rangle, 
\label{ec3}
\end{equation}
where the map $\langle \cdot,\cdot \rangle : \Omega^p(E) \times \mathfrak{X}(E) \to \Omega^{p - 1}(E)$ is defined by $\langle \Omega, v \rangle  := \iota_v(\Omega)$, for all $\Omega \in \Omega^p(E), v \in \mathfrak{X}(E)$\footnote{Note that $\Omega^0(E) = C^{\infty}(M, \mathbb{R})$, as in the usual case.}. A usual vector bundle connection $\tilde{\nabla}$ on $E$ induces a linear $E$-connection by $\nabla_u v := \tilde{\nabla}_{\rho(u)} v$. Hence, the existence of vector bundle connections implies that linear $E$-connections exist. 

\begin{Proposition} The set of all linear $E$-connections is an affine space $\mathfrak{C}$ modeled on $Tens^{(1, 2)}(E)$.
\label{pc1}
\end{Proposition}
\begin{proof} Due to the $E$-connection independent term $\rho(u)(f) v$ in the definition (\ref{ec2}) of a linear $E$-connection, the difference between two linear $E$-connections $\nabla, \nabla'$ is $C^{\infty}(M, \mathbb{R})$-bilinear:
\begin{equation} \Delta(\nabla, \nabla')(u, v) := \nabla_u v - \nabla'_u v,
\label{ec4}
\end{equation}
for all $u, v \in \mathfrak{X}(E)$, so that one can define the difference $E$-tensor as a $(1, 2)$-type $E$-tensor, which will be denoted by the same symbol
\begin{align} \Delta(\nabla, \nabla')(\Omega, u, v) &:= \langle \Omega, \Delta(\nabla, \nabla')(u, v) \rangle, \nonumber\\
\Delta(\nabla, \nabla')^a_{\ b c} &= \Gamma(\nabla)^a_{\ b c} - \Gamma(\nabla')^a_{\ b c}.
\label{ec5}
\end{align}
\end{proof}

\noindent Similar to the usual case (\ref{eb6}), the action of a linear $E$-connection can be also extended to all $E$-tensors by the Leibniz rule, and it induces a map $\nabla: Tens^{(q, r)}(E) \to Tens^{(q, r + 1)}(E)$, which is called the $E$-exterior covariant derivative . The only difference is that the second line in the equation (\ref{eb6}) should start with $\rho(u)$. 

\begin{Definition} The $E$-non-metricity tensor\footnote{A better name would be $E$-non-metricity $E$-tensor, but we will omit the second $E$ for aesthetic reasons.} corresponding to a linear $E$-connection $\nabla$ and an $E$-metric $g$ is defined as the $(0, 3)$-type $E$-tensor 
\begin{equation}
Q(\nabla, g) := \nabla g.
\label{ec6}
\end{equation}
\label{dc3}
\end{Definition}
\noindent If $Q(\nabla, g) = 0$, then $\nabla$ is called $E$-metric-$g$-compatible. Every vector bundle admits usual fiber-wise metric-compatible vector bundle connections, so the corresponding induced ones define $E$-metric-compatible linear $E$-connections. On a local $E$-frame $(X_a)$, components of the $E$-non-metricity tensor read
\begin{equation} Q(\nabla, g)_{a b c} = \rho(X_a)(g_{b c}) - \Gamma(\nabla)^d_{\ a b} g_{d c} - \Gamma(\nabla)^d_{\ a c}g_{b d}.
\label{ec7}
\end{equation}

For the constructions such as torsion and curvature operators, one needs to introduce a bracket on $E$-vector fields. A bracket on a vector bundle $E$ is an $\mathbb{R}$-bilinear map $[\cdot,\cdot]_E: \mathfrak{X}(E) \times \mathfrak{X}(E) \to \mathfrak{X}(E)$. For a vector bundle $E$ with a bracket $[\cdot,\cdot]_E$, $E$-anholonomy coefficients $\{ \gamma^a_{\ b c} \}$ can be defined identical to the usual case: $\gamma^a_{\ b c} := \langle e^a, [X_b, X_c]_E \rangle$ over a local $E$-frame $(X_a)$. For an anchored vector bundle $(E, \rho)$ with a bracket $[\cdot,\cdot]_E$, the pseudo-$E$-torsion map $T^{(0)}(\nabla): \mathfrak{X}(E) \times \mathfrak{X}(E) \to \mathfrak{X}(E)$ and the pseudo-$E$-curvature map $R^{(0)}(\nabla): \mathfrak{X}(E) \times \mathfrak{X}(E) \times \mathfrak{X}(E) \to \mathfrak{X}(E)$ of a linear $E$-connection $\nabla$ are defined by
\begin{equation} T^{(0)}(\nabla)(u, v) := \nabla_u v - \nabla_v u - [u, v]_E,
\label{ec8}
\end{equation}
\begin{equation} R^{(0)}(\nabla)(u, v, w) := \nabla_u \nabla_v w - \nabla_v \nabla_u w - \nabla_{[u, v]_E} w,
\label{ec9}
\end{equation}
for all $u, v, w \in \mathfrak{X}(E)$. 

In order to define $E$-tensorial objects from these maps, $C^{\infty}(M, \mathbb{R})$-linearity in each component is required. Hence, one needs to know how the $C^{\infty}(M, \mathbb{R})$-module structure is affected by the bracket. 
\begin{Definition} An (right) almost-Leibniz algebroid over $M$ is a triplet $(E, \rho, [\cdot,\cdot]_E)$ where $(E, \rho)$ is an anchored real vector bundle over $M$, $[\cdot,\cdot]_E$ is a bracket on $E$ satisfying the right-Leibniz rule
\begin{equation} [u, f v]_E = \rho(u)(f) v + f [u, v]_E,
\label{ec10}
\end{equation}
for all $u, v \in \mathfrak{X}(E), f \in C^{\infty}(M, \mathbb{R})$.
\label{dc4}
\end{Definition}

\noindent In this case, one can define ``Leibniz derivative''\footnote{Some authors call it Dorfman derivative or $E$-Lie derivative.} of an $E$-tensor completely analogous to the usual Lie bracket and Lie derivative relation. One should also deal with the first entry of the bracket.

\begin{Definition} A local almost-Leibniz algebroid is a quadruplet $(E, \rho, [\cdot,\cdot]_E, L)$ where $(E, \rho, [\cdot,\cdot]_E)$ is an almost-Leibniz algebroid, and $L: \Omega^1(E) \times \mathfrak{X}(E) \times \mathfrak{X}(E) \to \mathfrak{X}(E)$ is a $C^{\infty}(M, \mathbb{R})$-multilinear map, called the locality operator, satisfying the left-Leibniz rule \cite{14}, \cite{9}
\begin{equation} [f u, v]_E = - \rho(v)(f) u + f [u, v]_E + L(Df, u, v),
\label{ec11}
\end{equation}
for all $u, v \in \mathfrak{X}(E), f \in C^{\infty}(M, \mathbb{R})$, where $D: C^{\infty}(M, \mathbb{R}) \to \Omega^1(E)$ is the coboundary map defined by $(Df)(u) := \rho(u)(f)$ \cite{15}\footnote{Such a bracket will be called a local almost-Leibniz bracket, and this naming will be done for other type of brackets too.}.
\label{dc5}
\end{Definition}

\noindent Note that $D$ is related to usual exterior derivative $d: C^{\infty}(M, \mathbb{R}) \to \Omega^1(M)$ by $D = \rho^* \circ d$, where $\rho^*: T^*(M) \to E^*$ is the coanchor defined as the transpose of $\rho$, i. e. $\rho^*(\omega)(u) := \rho(u)(\omega)$, for all $\omega \in \Omega^1(M), u \in \mathfrak{X}(E)$.

Since $L$ is $C^{\infty}(M, \mathbb{R})$-multilinear, the $(2, 2)$-type locality $E$-tensor can be defined as
\begin{equation} L(\Omega, \Upsilon, u, v) := \langle \Omega, L(\Upsilon, u, v) \rangle,
\label{ec12}
\end{equation}
for all $\Omega, \Upsilon \in \Omega^1(E), u, v \in \mathfrak{X}(E)$. Note that in the equation (\ref{ec11}), one only deals with the $E$-exterior 1-forms which are in the image of the coboundary map $D$. This leads us to the following definition:

\begin{Definition} On an anchored vector bundle $(E, \rho)$, any two $C^{\infty}(M, \mathbb{R})$-multilinear maps $L, \tilde{L}: \Omega^1(E) \times \mathfrak{X}(E) \times \mathfrak{X}(E) \to \mathfrak{X}(E)$ will be called locally equivalent if
\begin{equation} L(Df, u, v) = \tilde{L}(Df, u, v), \nonumber
\end{equation}
for all $f \in C^{\infty}(M, \mathbb{R})$, $u, v \in \mathfrak{X}(E)$. Clearly, this is an equivalence relation, and an equivalence class $[\tilde{L}]$ will be called a locality structure on $(E, \rho)$.
\label{dc6}
\end{Definition}

\noindent Any two maps from the same locality structure define the same left-Leibniz rule, trivially. Hence, we will denote a local almost-Leibniz algebroid also by $(E, \rho, [\cdot,\cdot]_E, [L])$. Locality structures allow one to relate linear $E$-connections and local almost-Leibniz algebroid brackets. 

\begin{Proposition} Let $[\tilde{L}]$ be a locality structure on an anchored vector bundle $(E, \rho)$. Then the set of all local almost-Leibniz brackets corresponding to that locality structure defines an affine space $^{[\tilde{L}]} \mathfrak{B}$ modeled on $Tens^{(1, 2)}(E)$, and there is an affine map $^N \mathfrak{T}: \mathfrak{C} \to \ ^{[\tilde{L}]} \mathfrak{B}, \nabla \mapsto \ ^N [\cdot,\cdot]_{\nabla}$ defined by
\begin{equation} ^N [u, v]_{\nabla} := \nabla_u v - \nabla_v u + N(\nabla, L)(u, v),
\label{ec13}
\end{equation}
for any $\mathbb{R}$-bilinear map $N(\nabla, L): \mathfrak{X}(E) \times \mathfrak{X}(E) \to \mathfrak{X}(E)$ satisfying for any $L \in [\tilde{L}]$,
\begin{align} N(\nabla, L)(f u, v) &= L(Df, u, v) + f N(\nabla, L)(u, v), \nonumber\\
N(\nabla, L)(u, f v) &= f N(\nabla, L)(u, v),
\label{ec14}
\end{align}
for all $u, v \in \mathfrak{X}(E), f \in C^{\infty}(M, \mathbb{R})$.
\label{pc2}
\end{Proposition}
\begin{proof} Let $[\cdot,\cdot]_E$ and $[\cdot,\cdot]'_E$ be two local almost-Leibniz brackets with the locality operator $L, L'$ which are locally equivalent. Then as $L$ and $L'$ are locally equivalent
\begin{align} [f u, v]_E - [f u, v]'_E &= \left\{ - \rho(v)(f) u + f [u, v]_E + L(Df, u, v) \right\} \nonumber\\
& \qquad - \left\{ - \rho(v)(f) u + f [u, v]'_E + L'(Df, u, v) \right\} \nonumber\\
&= f \left\{ [u, v]_E - [u, v]'_E \right\}, \nonumber
\end{align}
for all $f \in C^{\infty}(M, \mathbb{R}), u, v \in \mathfrak{X}(E)$. Similarly, it can be shown that the difference is $C^{\infty}(M, \mathbb{R})$-linear in $v$. Hence, $^{[\tilde{L}]} \mathfrak{B}$ is an affine space modeled on $Tens^{(1, 2)}(E)$. Now, let $\nabla$ be a linear $E$-connection
\begin{align} ^N \mathfrak{T}(\nabla)(f u, v) &= \nabla_{f u} v - \nabla_v (f u) + N(\nabla, L)(f u, v) \nonumber\\
&= f \nabla_u v - \left[ \rho(v)(f) u + f \nabla_v u \right] + L(Df, u, v) + f N(\nabla, L)(u, v) \nonumber\\
&= - \rho(v)(f) u + f \ ^N \mathfrak{T}(\nabla)(u, v) + L(Df, u, v), \nonumber
\end{align}
which is the same as (\ref{ec11}) by the assumption (\ref{ec14}) on $N(\nabla, L)$. Similarly, it can be shown that it satisfies (\ref{ec10}) for $^N \mathfrak{T}(\nabla)(u, f v)$. Hence $^N \mathfrak{T}(\nabla)$ is a map from $\mathfrak{C}$ to $^{[\tilde{L}]} \mathfrak{B}$. One can directly check that it is affine by using similar arguments.
\end{proof}

\noindent Similarly, one can explicitly check that $T^{(0)}(\nabla)$ is not $C^{\infty}(M, \mathbb{R})$-linear in the first entry. As in the proposition (\ref{pc2}), locality structures allow one to modify pseudo-$E$-torsion map to have $C^{\infty}(M, \mathbb{R})$-bilinearity.

\begin{Corollary} Given a locality structure $[\tilde{L}]$ on a local almost-Leibniz algebroid, then for any $N(\nabla, L)$ satisfying (\ref{ec14})
\begin{equation} ^N \mathcal{T}(\nabla)(u, v) := T^{(0)}(\nabla)(u, v) + N(\nabla, L)(u, v)
\label{ec15}
\end{equation}
is $C^{\infty}(M, \mathbb{R})$-bilinear for any $L \in [\tilde{L}]$.
\label{cc1}
\end{Corollary}
\begin{proof} It is actually same thing as the proposition (\ref{pc2}) stated differently.
\end{proof}

\begin{Corollary} Given a locality structure $[\tilde{L}]$ on a local almost-Leibniz algebroid, then 
\begin{equation} ^L T(\nabla)(u, v) := T^{(0)}(\nabla)(u, v) + L(e^a, \nabla_{X_a} u, v)
\label{ec16}
\end{equation}
is $C^{\infty}(M, \mathbb{R})$-bilinear for any $L \in [\tilde{L}]$, where $(X_a)$ is a local $E$-frame \cite{9}.
\label{cc2}
\end{Corollary}
\begin{proof} By the fact that $\rho(X_a)(f) e^a = (Df)(X_a) e^a = Df$ for all $f \in C^{\infty}(M, \mathbb{R})$, one can explicitly check that for any $L \in [\tilde{L}], L(e^a, \nabla_{X_a} u, v)$ satisfies the necessary properties in (\ref{ec14}), and this makes $^L T(\nabla)(u, v)$ in the equation (\ref{ec16}) $C^{\infty}(M, \mathbb{R})$-bilinear.
\end{proof}

\noindent This is actually valid for $L(e^a, \nabla'_{X_a} u, v)$ for any linear $E$-connection $\nabla'$, but using the same linear $E$-connection $\nabla$ is more natural. One should note that the extra term's first entry is $e^a$, which can be outside of the image of the coboundary operator. Moreover, with this notation $T^{(0)}(\nabla) = \ ^0 T(\nabla)$.
\begin{Definition} Given a fixed locality structure representative $L$ on a local almost-Leibniz algebroid, the operator $^L T(\nabla): \mathfrak{X}(E) \times \mathfrak{X}(E) \to \mathfrak{X}(E)$ defined by (\ref{ec16}) is called the $E$-torsion operator of the linear $E$-connection $\nabla$, and the $E$-torsion tensor is defined as the $(1, 2)$-type $E$-tensor
\begin{equation} ^L T(\nabla)(\Omega, u, v) := \langle \Omega, \ ^L T(\nabla)(u, v) \rangle,
\label{ec17}
\end{equation}
for all $\Omega \in \Omega^1(E), u, v \in \mathfrak{X}(E)$. 
\label{dc7}
\end{Definition}

\noindent If $^L T(\nabla)(u, v) = 0$, for all $u, v \in \mathfrak{X}(E)$, then $\nabla$ is called $E$-torsion-free. On a local $E$-frame $(X_a)$, components of the $E$-torsion tensor read
\begin{equation} ^L T(\nabla)^a_{\ b c} = \Gamma(\nabla)^a_{\ b c} - \Gamma(\nabla)^a_{\ c b} - \gamma^a_{\ b c} + \Gamma(\nabla)^e_{\ d b} L^{a d}_{\ \ e c}.
\label{ec18}
\end{equation}

Modification of the pseudo-$E$-curvature map to have $C^{\infty}(M, \mathbb{R})$-multilinearity requires more assumptions. First, $\rho$ should be a morphism of anchored vector bundles with a bracket, i. e.
\begin{equation} \rho([u, v]_E) = [\rho(u), \rho(v)],
\label{ec19}
\end{equation}
for all $u, v \in \mathfrak{X}(E)$\footnote{And also, $id_{T(M)} \circ \rho = \rho$, which is trivially satisfied.}. 
\begin{Definition} An almost-Leibniz algebroid whose anchor and bracket satisfy (\ref{ec19}) is called a pre-Leibniz algebroid \cite{16}.
\label{dc8}
\end{Definition}

A sufficient condition for an almost-Leibniz algebroid $(E, \rho, [\cdot,\cdot]_E$) to be a pre-Leibniz algebroid is that $(\mathfrak{X}(E), [\cdot,\cdot]_E)$ forms a (right) Leibniz algebra, i. e. it should satisfy the (right) Leibniz identity
\begin{equation} [u, [v, w]_E]_E = [[u, v]_E, w]_E + [v, [u, w]_E]_E,
\label{ec20}
\end{equation}
for all $u, v, w \in \mathfrak{X}(E)$. In this case, $(E, \rho, [\cdot,\cdot]_E)$ is called a (right) Leibniz algebroid. On the other hand, this condition is not necessary; for example, every almost-Leibniz algebroid with $\rho = 0$ is a pre-Leibniz algebroid, trivially. Every Leibniz algebroid satisfies the following property due to (\ref{ec20})
\begin{equation} [[u, v]_E + [v, u]_E, w]_E = 0,
\label{ec21}
\end{equation}
for all $u, v, w \in \mathfrak{X}(E)$. Moreover, for a local Leibniz algebroid $E$ over $M$, there is a natural Leibniz algebroid structure on $E|_U$ for every open $U \subset M$ \cite{17}, so that the name, locality, is well-chosen.

\begin{Proposition} Given a locality structure $[\tilde{L}]$ on a local pre-Leibniz algebroid, the following map is $C^{\infty}(M, \mathbb{R})$-multilinear:
\begin{equation} ^M \mathcal{R}(\nabla)(u, v, w) := \nabla_u \nabla_v w - \nabla_v \nabla_u w - \nabla_{[u, v]_E} w + M(\nabla, L)(u, v, w),
\label{ec22}
\end{equation}
for any $\mathbb{R}$-multilinear map $M(\nabla, L): \mathfrak{X}(E) \times \mathfrak{X}(E) \times \mathfrak{X}(E) \to \mathfrak{X}(E)$ satisfying for any $L \in [\tilde{L}]$
\begin{align} M(\nabla, L)(f u, v, w) &= \nabla_{L(Df, u, v)} w + f M(\nabla, L)(u, v, w), \nonumber\\
M(\nabla, L)(u, f v, w) &= M(\nabla, L)(u, v, f w) = f M(\nabla, L)(u, v, w),
\label{ec23}
\end{align}
for all $u, v, w \in \mathfrak{X}(E), f \in C^{\infty}(M, \mathbb{R})$.
\label{pc3}
\end{Proposition}
\begin{proof} $C^{\infty}(M, \mathbb{R})$-linearity conditions for $u$ and $v$ are similar for the previous propositions. 
\begin{align} ^M \mathcal{R}(\nabla)(u, v, f w) &= \nabla_u \nabla_v (f w) - \nabla_v \nabla_u (f w) - \nabla_{[u, v]_E} (f w) + M(u, v, f w)\nonumber\\
&= \nabla_u \left[ \rho(v)(f) w + f \nabla_v w \right] - \nabla_v \left[ \rho(u)(f) w + f \nabla_u w \right] \nonumber\\
& \qquad - \rho([u, v]_E)(f) w + f \nabla_{[u, v]_E} w + f M(u, v, w) \nonumber\\
&= \left\{ \rho(u)(\rho(v)(f)) - \rho(v)(\rho(u))(f) - \rho([u, v]_E)(f) \right\} w \nonumber\\
& \qquad + f \ ^M \mathcal{R}(\nabla)(u, v, w)  + \rho(u)(f) \nabla_v w - \rho(u)(f) \nabla_v w \nonumber\\
& \qquad +  \rho(v)(f) \nabla_u w - \rho(v)(f) \nabla_u w \nonumber\\
&= \left\{ [\rho(u), \rho(v)](f) - \rho([u, v]_E)(f) \right\} w + f \ ^M \mathcal{R}(\nabla)(u, v, w) \nonumber\\
&= f \ ^M \mathcal{R}(\nabla)(u, v, w) \nonumber
\end{align}
for all $f \in C^{\infty}(M, \mathbb{R}), u, v, w \in \mathfrak{X}(E)$, where the last step follows from the assumption (\ref{ec19}) for pre-Leibniz algebroids.
\end{proof}

\noindent In the following statements, the kernel of the anchor map $\rho$ will be relevant. Vector bundles do not form an abelian category, meaning that kernels and cokernels of vector bundle morphisms do not have to be vector bundles themselves. The kernel of a vector bundle morphism $\psi: E \to F$ is a subbundle if and only if $\psi$ is of locally constant rank, i. e. $ker(\psi)$ has a locally constant rank. 

\begin{Definition} An anchored vector bundle $(E, \rho)$ is said to be regular if $\rho$ is of locally constant rank.
\label{dc9}
\end{Definition}

\begin{Proposition} On a regular anchored vector bundle $(E, \rho)$, let $N(\nabla, L)$ be as in (\ref{ec14}) with $im(N) \subset ker(\rho)$, then $\nabla_{N(\nabla, L)(u, v)} w$ satisfies the properties (\ref{ec23}).
\label{pc4}
\end{Proposition}
\begin{proof} By the properties (\ref{ec2}) of $E$-connections, (\ref{ec14}) of $N(\nabla, L)$
\begin{equation}
\nabla_{N(\nabla, L)(f u, v)} w = \nabla_{L(Df, u, v) + f N(\nabla, L)(u, v)} w = \nabla_{L(Df, u, v)} w + f \nabla_{N(\nabla, L)(u, v)} w, \nonumber
\end{equation}
\begin{equation}
\nabla_{N(\nabla, L)(u, f v)} w = \nabla_{f N(\nabla, L)(u, v)} w = f \nabla_{N(\nabla, L)(u, v)} w, \nonumber
\end{equation}
\begin{equation}
\nabla_{N(\nabla, L)(u, v)} (f w) = \rho(N(\nabla, L)(u, v)) w + f \nabla_{N(\nabla, L)(u, v)} w = f \nabla_{N(\nabla, L)(u, v)} w, \nonumber
\end{equation}
where the last step in the last equation follows from the assumption $im(N) \subset ker(\rho)$.
\end{proof}

\noindent For a local pre-Leibniz algebroid, one has $(\rho \circ L)(Df, u, v) = 0$, for all $u, v \in \mathfrak{X}(E), f \in C^{\infty}(M, \mathbb{R})$. If one can find a way to make this true outside of the image of the coboundary map $D$, then by the proposition (\ref{pc4}), $L(e^a, \nabla_{X_a} u, v)$ would be helpful to define $E$-curvature map. With the assumption that $(E, \rho)$ is regular, the necessary structure for this purpose is a ``projector'' which maps $[\tilde{L}]$ to a sub-equivalence class of $[\tilde{L}]$ whose elements' images are in $ker(\rho)$. 
\begin{Definition} Let $[L]$ be a locality structure on a regular anchored vector bundle $(E, \rho)$. A locality projector is defined to be a $C^{\infty}(M, \mathbb{R})$-linear map $\mathcal{P}: \mathfrak{X}(E) \to \mathfrak{X}(E)$ satisfying $\hat{L} := \mathcal{P} \circ L \in [\tilde{L}]$ for all $L \in [\tilde{L}]$ such that $im(\hat{L}) \subset ker(\rho)$. It is $C^{\infty}(M, \mathbb{R})$-linear, so one can define a $(1, 1)$-type locality projector $E$-tensor as $P(\Omega, u) := \langle \Omega, P(u) \rangle$, for all $\Omega \in \Omega^1(E), u \in \mathfrak{X}(E)$.
\label{dc10}
\end{Definition}

\begin{Corollary} Given a locality structure $[\tilde{L}]$ with a locality projector $\mathcal{P}$ on a regular local pre-Leibniz algebroid,
\begin{equation} ^{\hat{L}} R(\nabla)(u, v, w) := R^{(0)}(\nabla)(u, v, w) + \nabla_{\hat{L}(e^a, \nabla_{X_a} u, v)} w,
\label{ec24}
\end{equation}
is $C^{\infty}(M, \mathbb{R})$-multilinear for any $L \in [\tilde{L}]$, where $\hat{L} = \mathcal{P} \circ L$, and $(X_a)$ is a local $E$-frame. 
\label{cc3}
\end{Corollary}
\begin{proof} Direct consequence of the corollary (\ref{cc2}) and propositions (\ref{pc3}, \ref{pc4}).
\end{proof}

\begin{Definition} Given a fixed locality structure representative $L$ on a regular local pre-Leibniz algebroid with a locality projector $\mathcal{P}$, the operator $^{\hat{L}} R(\nabla): \mathfrak{X}(E) \times \mathfrak{X}(E) \times \mathfrak{X}(E) \to \mathfrak{X}(E)$ defined by (\ref{ec24}) is called the $E$-curvature operator of the linear $E$-connection $\nabla$, and the $E$-curvature tensor is defined as a $(1, 3)$-type $E$-tensor
\begin{equation}
^{\hat{L}} R(\nabla)(\Omega, u, v, w) := \langle \Omega, \ ^{\hat{L}} R(\nabla)(u, v, w) \rangle,
\label{ec25}
\end{equation}
for all $\Omega \in \Omega^1(E), u, v, w \in \mathfrak{X}(E)$. 
\label{dc11}
\end{Definition}

\noindent On a local $E$-frame $(X_a)$, the $E$-curvature components read
\begin{align} ^{\hat{L}} R(\nabla)^a_{\ b c d} = & \ \rho(X_b) \left( \Gamma(\nabla)^a_{\ c d} \right) - \rho(X_c) \left( \Gamma(\nabla)^a_{\ b d} \right) + \Gamma(\nabla)^e_{\ c d} \Gamma(\nabla)^a_{\ b e} \nonumber\\
& - \Gamma(\nabla)^e_{\ b d} \Gamma(\nabla)^a_{\ c e} - \gamma^e_{\ b c} \Gamma(\nabla)^a_{\ e d} + \Gamma(\nabla)^f_{\ e b} \Gamma(\nabla)^a_{\ g d} \hat{L}^{g e}_{\ \ f c},
\label{ec26}
\end{align}
where $\hat{L}^{a b}_{\ \ c d} = L^{a e}_{\ \ c d} P^b_{\ e}$. $E$-Ricci tensor and $E$-Ricci scalar are defined as in the usual case, and
\begin{equation} ^{\hat{L}} Ric(\nabla)_{a b} = \ ^{\hat{L}} R(\nabla)^c_{\ c a b},
\label{ec27}
\end{equation}
\begin{equation} ^{\hat{L}} R(\nabla, g) = \ ^{\hat{L}} Ric(\nabla)_{a b} g^{a b},
\label{ec28}
\end{equation}
on a local $E$-frame $(X_a)$ \cite{9}. 

\noindent With all these structure, we are ready to give the main definition of this paper.

\begin{Definition} $E$-metric-connection geometries are defined as a quadruplet $(M, (E, \rho,$ $[\cdot,\cdot]_E, [L], \mathcal{P}), g, \nabla)$, where $M$ is a manifold, $(E, \rho, [\cdot,\cdot]_E, [L], \mathcal{P})$ is a regular local pre-Leibniz algebroid over $M$ with a locality projector $\mathcal{P}$ on the locality structure $[L]$, $g$ is an $E$-metric, and $\nabla$ is a linear $E$-connection. 
\label{dc12}
\end{Definition}

\noindent In order to simplify the constructions, the same representative of $[L]$ will be chosen for all structures. Moreover, if a pre-Leibniz algebroid is denoted by a unique locality operator $L$, then it should be understood that the structures are defined with respect to this $L$.

One can prove that $E$-metric-connection geometries ``generalize'' metric-affine geometries:

\begin{Theorem} $(M, (T(M), id_{T(M)}, [\cdot,\cdot], [0], 0), g, \nabla)$ induces a unique metric-affine geometry on $M$.
\label{tc1}
\end{Theorem}
\begin{proof} When one considers $T(M)$ as the vector bundle in the constructions, $T(M)$-tensors are automatically the usual tensors; in particular this applies for the $T(M)$-metric $g$. Due to the fact that the anchor is the identity map $id_{T(M)}$, linear $T(M)$-connections are the usual affine connections. This makes $T(M)$-non-metricity tensors coincide with the usual ones. The Lie bracket $[\cdot,\cdot]$ satisfies the necessary conditions for a local almost-Leibniz bracket with $L = 0$ as it is anti-symmetric (see the proof of the proposition \ref{pc10})). Hence, $(T(M), id_{T(M)}, [\cdot,\cdot], [0])$ is a local almost-Leibniz algebroid. Since, $id_{T(M)}$ is the anchor, it is actually a regular pre-Leibniz algebroid. Moreover, $T(M)$-torsion and $T(M)$-curvature operators coincide with the usual definitions for $L = 0, \mathcal{P} = 0$. Yet, these operators could have been defined in terms of any $0 \neq L \in [0]$. One can show that such $L$ does not exist as the coboundary map $D$ coincides with the exterior derivative $d$ acting on smooth functions: $Df(U) = id_{T(M)}(U)(f) = U(f) = df(U)$ for $f \in C^{\infty}(M, \mathbb{R}), U \in \mathfrak{X}(M)$, so that $D = d$. Let $L \in [0]$ be distinct from 0, then because $L$ and 0 are locally equivalent and $D = d$, one has $L(Df, U, V) = L(df, u, v) = 0(df, u, v) = 0$, for all $f \in C^{\infty}(M, \mathbb{R}), U, V \in \mathfrak{X}(M)$. Remembering the additional term $L(e^a, \nabla_{X_a} u, v)$ for $E$-torsion and $E$-curvature, one needs to show that $L$ should be equal to 0 outside of the image of the exterior derivative, i. e. on non-exact 1-forms. On a local trivialization chart, any 1-form $\omega$ can be written as $\omega = \omega_i dx^i$ for some $\{ \omega_i \} \subset C^{\infty}(M, \mathbb{R})$ as $(dx^i)$ forms a local frame for $T^*(M)$. Hence, as $L$ is $C^{\infty}(M, \mathbb{R})$-multilinear
\begin{equation} L(\omega, U, V) = L(\omega_i dx^i, U, V) = \omega_i L(dx^i, U, V) = \sum_{i = 1}^n (\omega_i 0) = 0, \nonumber
\end{equation}
which is a contradiction. Therefore, one has $[0] = \{ 0 \}$ so that this generalization is unique.
\end{proof}

\begin{Proposition} Given an $E$-metric $g$ and two linear $E$-connections $\nabla$ and $\nabla'$, their $E$-non-metricity, $E$-torsion and $E$-curvature components are related by
\begin{align}  \nonumber\\
Q(\nabla, g)_{a b c} &= Q(\nabla', g)_{a b c} - \Delta(\nabla, \nabla')^d_{\ a b} g_{d c} - \Delta(\nabla, \nabla')^d_{\ a c} g_{b d}, \nonumber\\
^L T(\nabla)^a_{\ b c} &= \ ^L T(\nabla')^a_{\ b c} + \Delta(\nabla, \nabla')^a_{\ b c} - \Delta(\nabla, \nabla')^a_{\ c b} + \Delta(\nabla, \nabla')^d_{\ e b} L^{a e}_{\ \ d c},
\label{ec29}
\end{align}
\begin{align} ^{\hat{L}} R(\nabla)^a_{\ b c d} = & \ ^{\hat{L}} R(\nabla')^a_{\ b c d} + \ ^{\hat{L}} R(\Delta(\nabla, \nabla'))^a_{\ b c d} \nonumber\\
& + \Gamma(\nabla')^e_{\ c d} \Delta(\nabla, \nabla')^a_{\ b e} - \Gamma(\nabla')^e_{\ b d} \Delta(\nabla, \nabla')^a_{\ c e} \nonumber\\
& + \Gamma(\nabla')^a_{\ b e} \Delta(\nabla, \nabla')^e_{\ c d} - \Gamma(\nabla')^a_{\ c e} \Delta(\nabla, \nabla')^e_{\ b d} \nonumber\\
& + \Gamma(\nabla')^f_{\ e b} \Delta(\nabla, \nabla')^a_{\ g d} \hat{L}^{g e}_{\ \ f c} + \Gamma(\nabla')^a_{\ g d} \Delta(\nabla, \nabla')^f_{\ e b} \hat{L}^{g e}_{\ \ f c},
\label{ec30}
\end{align}
where $\{ ^{\hat{L}} R(\Delta(\nabla, \nabla'))^a_{\ b c d} \}$ are defined in a similar way to $E$-curvature tensor components but in terms of the components of the difference $E$-tensor instead of the $E$-connection coefficients.
\label{pc5}
\end{Proposition}
\begin{proof} It is a trivial result of the equations (\ref{ec5}, \ref{ec7}, \ref{ec18} and \ref{ec27}).
\end{proof}

\begin{Definition} A linear $E$-connection $\nabla$ is called an $E$-Levi-Civita connection corresponding to an $E$-metric $g$ if it is $E$-torsion-free and $E$-metric-$g$-compatible.
\label{dc13}
\end{Definition}

\begin{Corollary} In order to make a linear $E$-connection $\nabla$ an $E$-Levi-Civita connection, the difference $E$-tensor components for any linear $E$-connection $\nabla'$ should satisfy
\begin{equation} \Delta(\nabla, \nabla')^d_{\ a b} g_{c d} + \Delta(\nabla, \nabla')^d_{\ a c} g_{b d} = Q(\nabla', g)_{a b c},
\label{ec31}
\end{equation}
\begin{equation} \Delta(\nabla, \nabla')^a_{\ b c} - \Delta(\nabla, \nabla')^a_{\ c b} + \Delta(\nabla, \nabla')^d_{\ e b} L^{a e}_{\ \ d c} = - \ ^L T(\nabla')^a_{\ b c}.
\label{ec32}
\end{equation}
\label{cc4}
\end{Corollary}
\begin{proof} These equations directly follow from the equation (\ref{ec29}).
\end{proof}

\begin{Proposition} $^N \mathfrak{T}: \mathfrak{C} \to \ ^{[\tilde{L}]} \mathfrak{B}$ defined as in the proposition (\ref{pc3}) for $N(\nabla, L)(u, v) = L(e^a, \nabla_{X_a} u, v)$ for a local $E$-frame $(X_a)$ is an affine surjection if 
\begin{equation} L(e^a, Z(X_a, u), v) = Z(v, u) \pm Z(u, v),
\label{ec33}
\end{equation}
for all $C^{\infty}(M, \mathbb{R})$-bilinear maps $Z: \mathfrak{X}(E) \times \mathfrak{X}(E) \to \mathfrak{X}(E), u, v \in \mathfrak{X}(M)$ and for all local $E$-frames $(X_a)$. For both cases, $E$-Levi-Civita connections always exist.
\label{pc6}
\end{Proposition}
\begin{proof} Let $\tilde{\nabla}$ be a linear $E$-connection with the corresponding local almost-Leibniz bracket $^N[\cdot,\cdot]_{\tilde{\nabla}}$, and $[\cdot,\cdot]'$ any local almost-Leibniz bracket. By the proposition (\ref{pc2}), $^{[\tilde{L}]} \mathfrak{B}$ is an affine space modeled on $Tens^{(1, 2)}(E)$ so that there exits a $C^{\infty}(M, \mathbb{R})$-bilinear map $\mathcal{K}: \mathfrak{X}(E) \times \mathfrak{X}(E) \to \mathfrak{X}(E)$ such that
\begin{equation} [u, v]' - \ ^N [u, v]_{\tilde{\nabla}} = \mathcal{K}(u, v), \nonumber
\end{equation}
for all $u, v \in \mathfrak{X}(E)$. $\mathfrak{C}$ is also an affine space modeled on $Tens^{(1, 2)}(E)$ so as in the equation (\ref{ec5}), for any linear $E$-connection $\nabla$ 
\begin{equation} \nabla - \tilde{\nabla} = \Delta(\nabla, \tilde{\nabla}) \nonumber
\end{equation}
defines a $(1, 2)$-type $E$-tensor. First, consider $L(e^a, Z(X_a, u), v) = Z(v, u) + Z(u, v)$ case:
\begin{align} 
^N \mathfrak{T}(\nabla)(u, v) &= \ ^N [u, v]_{\nabla} \nonumber\\
&= \nabla_u v - \nabla_v u + L(e^a, \nabla_{X_a} u, v) \nonumber\\
&= \tilde{\nabla}_u v + \Delta(\nabla, \tilde{\nabla})(u, v) - \tilde{\nabla}_v u - \Delta(\nabla, \tilde{\nabla})(v, u) + L(e^a, \tilde{\nabla}_{X_a} u, v) \nonumber\\
& \qquad + L(e^a, \Delta(\nabla, \tilde{\nabla})(X_a, u), v) \nonumber\\
&= \ ^N [u, v]_{\tilde{\nabla}} + \Delta(\nabla, \tilde{\nabla})(u, v) - \Delta(\nabla, \tilde{\nabla})(v, u) \nonumber\\
& \qquad + \Delta(\nabla, \tilde{\nabla})(u, v) + \Delta(\nabla, \tilde{\nabla})(v, u) \nonumber\\
&= \ ^N [u, v]_{\tilde{\nabla}} + 2 \Delta(\nabla, \tilde{\nabla})(u, v), \nonumber
\end{align}
where the assumption (\ref{ec33}) is used. Hence, if one chooses
\begin{equation} \Delta(\nabla, \tilde{\nabla})(u, v) = \frac{1}{2} \mathcal{K}(u, v), \nonumber
\end{equation}
then $^N \mathfrak{T}(\nabla) = [\cdot,\cdot]'$ so that the map is surjective. Since the map is surjective, for every local almost-Leibniz bracket, there is a linear $E$-connection $\nabla$ such that $E$-torsion map defined by (\ref{ec16}) vanishes. Hence, $\nabla$ is $E$-torsion-free. This also means that if the assumption (\ref{ec33}) holds, then the equation (\ref{ec31}) is the decomposition of the $E$-torsion into its anti-symmetric and symmetric parts. 

For this case, one can define the $E$-contorsion tensor $C(\nabla, g)$ as an $E$-tensor whose components are given by
\begin{equation} C(\nabla, g)_{a b c} := T(\nabla)_{[a b c]}, \nonumber
\end{equation}
where the index lowering by the $E$-metric $g$ is used, and $[\cdots]$ denote the total anti-symmetrization of the indices. By direct computation one can show that for a metric-$g$-compatible linear $E$-connection $\nabla$, one can construct an $E$-Levi-Civita connection $\nabla'$ if $\Delta(\nabla, \nabla') = C(\nabla, g)$.

Similarly for the case $L(e^a, Z(X_a, u), v) = Z(v, u) - Z(u, v)$
\begin{equation} ^N \mathfrak{T}(\nabla)(u, v) = \ ^N [u, v]_{\nabla} = \ ^N [u, v]_{\tilde{\nabla}} - 2 \Delta(\nabla, \tilde{\nabla})(v, u), \nonumber
\end{equation}
Therefore, if one chooses 
\begin{equation} \Delta(\nabla, \tilde{\nabla})(u, v) = - \frac{1}{2} \mathcal{K}(v, u), \nonumber
\end{equation}
then $^N \mathfrak{T}(\nabla) = [\cdot,\cdot]'$ so that the map is surjective, and an $E$-torsion-free $E$-connection exists. This case yields a ``pathological'' fact about the $E$-torsion: Every linear $E$-connection has the same $E$-torsion by the equation (\ref{ec29}), so that every $E$-connection is $E$-torsion-free. In particular, every $E$-metric-compatible linear $E$-connections are $E$-torsion-free, so $E$-Levi-Civita connections exist.
\end{proof}

\noindent In general, $E$-Levi-Civita connections might not exist or when they exist they might not be unique, as opposed to the usual case. One can try to find a linear $E$-connection in the form analogous to the usual Koszul formula (\ref{eb17}), but it does not define a linear $E$-connection. Yet, similar to the pseudo-$E$-torsion map and pseudo-$E$-curvature map, one can modify the Koszul formula to have a linear $E$-connection.

\begin{Proposition} On a local almost Leibniz algebroid endowed with an $E$-metric $g$ and a locality structure $[\tilde{L}]$, the following modification of the Koszul formula (\ref{eb17}) defines a linear $E$-connection $\nabla$
\begin{align} 2 g(\nabla_u v, w) + \mathcal{K}(\nabla, g, L)(u, v, w) &= \rho(u)(g(v, w)) + \rho(v)(g(u, w)) - \rho(w)(g(u, v)) \nonumber\\
& \qquad - g([v, w]_E, u) - g([u, w]_E, v) + g([u, v]_E, w),
\label{ec34}
\end{align}
for any $\mathbb{R}$-multilinear map $\mathcal{K}(\nabla, g, L): \mathfrak{X}(E) \times \mathfrak{X}(E) \times \mathfrak{X}(E) \to \mathfrak{X}(E)$ satisfying for any $L \in [\tilde{L}]$
\begin{align}
\mathcal{K}(\nabla, g, L)(f u, v, w) &= - g(L(Df, u, w), v) + g(L(Df, u, v), w) + f \mathcal{K}(\nabla, g, L)(u, v, w), \nonumber\\
\mathcal{K}(\nabla, g, L)(u, f v, w) &= - g(L(Df, v, w), u) + f \mathcal{K}(\nabla, g, L)(u, v, w), \nonumber\\
\mathcal{K}(\nabla, g, L)(u, v, f w) &=  f \mathcal{K}(\nabla, g, L)(u, v, w),
\label{ec35}
\end{align}
for all $f \in C^{\infty}(M, \mathbb{R}), u, v, w \in \mathfrak{X}(E)$.
\label{pc7}
\end{Proposition}
\begin{proof} One should check that the defining properties of linear $E$-connections (\ref{ec2}) hold. By linearity of $g$ and the assumption (\ref{ec35})
\begin{align} 2 g(\nabla_{f u} v, w) + \mathcal{K}(\nabla, g, L)(f u, v, w) &= 2 g(\nabla_{f u} v, w) - g(L(Df, u, w), v) \nonumber\\
& \quad + g(L(Df, u, v), w) + f \mathcal{K}(\nabla, g, L)(u, v, w), \nonumber
\end{align}
\begin{align} 
& \Big\{ \rho(f u)(g(v, w)) + \rho(v)(g(f u, w)) - \rho(w)(g(f u, v)) - g([v, w]_E, f u) - g([f u, w]_E, v) \nonumber\\
& \quad \qquad + g([f u, v]_E, w) \Big\} \nonumber\\
& = f \rho(u)(g(v, w)) + \rho(v)(f g(u, w)) - \rho(w)(f g(u, v)) - f g([v, w]_E, u) \nonumber\\
& \quad - g(- \rho(w)(f)u + f[u, w]_E + L(Df, u, w), v) \nonumber\\
& \quad + g(- \rho(v)(f)u + f[u, v]_E + L(Df, u, v), w) \nonumber\\
& = f \rho(u)(g(v, w)) + \rho(v)(f) g(u, w) + f \rho(v)(g(u, w)) - \rho(w)(f) g(u, v) - f \rho(w)(g(u, v)) \nonumber\\
& \quad - f g([v, w]_E, u) + \rho(w)(f) g(u, v) - f g([u, w]_E, v) - g(L(Df, u, w), v) \nonumber\\
& \quad - \rho(v)(f) g(u, w) + f g([u, v]_E, w) + g(L(Df, u, v, w) \nonumber\\
& = f \Big\{ \rho(u)(g(v, w)) + \rho(v)(g(u, w)) - \rho(w)(g(u, v)) - g([v, w]_E, u) \nonumber\\
& \quad - g([u, w]_E, v) + g([u, v]_E, w) \Big\} - g(L(Df, u, w), v) + g(L(Df, u, v), w) \nonumber\\
& \quad + f \mathcal{K}(\nabla, g, L)(u, v, w) \nonumber
\end{align}
These two equations coincide if $\nabla_{f u} v = f \nabla_u v$ for all $f \in C^{\infty}(M, \mathbb{R}), u, v \in \mathfrak{X}(E)$, which is one of the needed properties. Similarly for the $v$ entry, by using the fact that $g$ is symmetric
\begin{align} 2 g(\nabla_u (f v), w) + \mathcal{K}(\nabla, g, L)(u, f v, w) &= 2 g(\nabla_u (f v), w) - g(L(Df, v, w), u) \nonumber\\
& f \mathcal{K}(\nabla, g, L)(u, v, w), \nonumber
\end{align}
\begin{align} 
& \Big\{ \rho(u)(g(f v, w)) + \rho(f v)(g(u, w)) - \rho(w)(g(u, f v)) - g([f v, w]_E, u) - g([u, w]_E, f v) \nonumber\\
& \quad \qquad + g([u, f v]_E, w) \Big\} \nonumber\\
& = \rho(u)(f g(v, w)) + f \rho(v)(g(u, w)) - \rho(w)(f g(u, v)) \nonumber\\
& \quad  - g(- \rho(w)(f)v + f[v, w]_E + L(Df, v, w), u) - f g([u, w]_E, v) \nonumber\\
& \quad + g(\rho(u)(f)v + f[u, v]_E, w) \nonumber\\
& = \rho(u)(f) g(v, w) + f \rho(u)(g(v, w)) + f \rho(v)(g(u, w)) - \rho(w)(f) g(u, v) - f \rho(w)(g(u, v)) \nonumber\\
& \quad + \rho(w)(f) g(v, u) - f g([v, w]_E, u) - g(L(Df, v, w), u) - f g([u, w]_E, v) \nonumber\\
& \quad + \rho(u)(f) g(v, w) + f g([u, v]_E, w)\nonumber\\
& = \rho(u)(f) g(v, w) + f \Big\{ \rho(u)(g(v, w)) + \rho(v)(g(u, w)) - \rho(w)(g(u, v)) - g([v, w]_E, u) \nonumber\\
& \quad - g([u, w]_E, v) + g([u, v]_E, w) \Big\} - g(L(Df, v, w), u) \nonumber
\end{align}
Similarly these two equations coincide if $\nabla_u (f v) = \rho(u)(f) v + f \nabla_u v$ for all $f \in C^{\infty}(M, \mathbb{R}), u, v \in \mathfrak{X}(E)$, which is the other defining property of linear $E$-connections. Moreover, one can also check for that both sides are $C^{\infty}(M, \mathbb{R})$-linear in $w$. Therefore, the equation (\ref{ec34}) defines a linear $E$-connection. 
\end{proof}

\begin{Proposition}
The following map $K(\nabla, g, L): \mathfrak{X}(E) \times \mathfrak{X}(E) \times \mathfrak{X}(E) \to \mathfrak{X}(E)$ satisfies the necessary conditions (\ref{ec35})
\begin{align} 
K(\nabla, g, L)(u, v, w) := & - g (L ( e^a, \nabla_{X_a} v, w ), u) - g (L ( e^a, \nabla_{X_a} u, w ), v) \nonumber\\
& + g (L ( e^a, \nabla_{X_a} u, v ), w),
\label{ec36}
\end{align}
for all $u, v, w \in \mathfrak{X}(E)$, where $(X_a)$ is a local $E$-frame.
\label{pc8}
\end{Proposition}
\begin{proof} Direct calculation by using the definition of a linear $E$-connection (\ref{ec2}) and multilinearity of $g$ and $L$.
\end{proof}

\begin{Definition} Given a fixed locality structure representative $L$ on a local almost-Leibniz algebroid, a linear $E$-connection, denoted by $^K \nabla$, will be called an $E$-Koszul connection if it satisfies the equation (\ref{ec34}) with $\mathcal{K}(^K \nabla, g, L) = K(^K \nabla, g, L)$, which is defined by the equation (\ref{ec36}).
\label{dc14}
\end{Definition}
\noindent Clearly, this modification is not useful as the usual Koszul formula due to the extra term $K(^K \nabla, g, L)$ so that one cannot directly compute its coefficients. Moreover, there is no reason to have a unique $E$-Koszul connection for an arbitrary locality operator. Nevertheless, $E$-Koszul connections have some interesting properties. On a local $E$-frame $(X_a)$, the equations (\ref{ec34})  and (\ref{ec36}) yield
\begin{align} \Gamma(^K \nabla)^a_{\ b c} = & \frac{1}{2} g^{a d} \Bigl[ \rho(X_b) \left( g_{c d} \right) + \rho(X_c) \left( g_{b d} \right) - \rho(X_d) \left( g_{b c} \right) \nonumber\\
& \qquad - \gamma^e_{\ c d} g_{e b} - \gamma^e_{\ b d} g_{e c} + \gamma^e_{\ b c} g_{e d} + \Gamma(^K \nabla)^e_{\ g c} L^{f g}_{\ \ e d} g_{f b} \nonumber\\
& \qquad + \Gamma(^K \nabla)^e_{\ g b} L^{f g}_{\ \ e d} g_{f c} - \Gamma(^K \nabla)^e_{\ g b} L^{f g}_{\ \ e c} g_{f d} \Bigr].
\label{ec37}
\end{align}

\begin{Proposition} For an $E$-Koszul connection $^K \nabla$, $E$-torsion and $E$-non-metricity components satisfy
\begin{equation} Q(^K \nabla, g)_{a b c} = - \ ^L T(^K \nabla)^f_{\ b c} g_{f a}.
\label{ec38}
\end{equation} 
\label{pc9}
\end{Proposition}
\begin{proof}
By the equation (\ref{ec37}), $E$-torsion and $E$-non-metricity components of an $E$-Koszul connection can be evaluated as
\begin{equation} ^L T(^K \nabla)^a_{\ b c} = \frac{1}{2} \left[ \Gamma(^K \nabla)^e_{\ d c} L^{a d}_{\ \ e b} + \Gamma(^K \nabla)^e_{\ d b} L^{a d}_{\ \ e c} - \gamma^a_{\ b c} - \gamma^a_{\ c b} \right],
\label{ec39}
\end{equation}
\begin{equation} Q(^K \nabla, g) = \frac{1}{2} \left[ \gamma^f_{\ b c} + \gamma^f_{\ c b} - \Gamma(^K \nabla)^e_{\ d c} L^{f d}_{\ \ e b} - \Gamma(^K \nabla)^e_{\ d b} L^{f d}_{\ \ e c}  \right] g_{f a}.
\label{ec40}
\end{equation}
Hence, they satisfy the equation (\ref{ec38}).
\end{proof}
\noindent This result shows that conditions for being $E$-torsion-free and $E$-metric-$g$-compatible are not independent from each other for an $E$-Koszul connection. Moreover, if an $E$-Koszul connection is $E$-torsion-free, then it is automatically an $E$-Levi-Civita connection. This hidden property is valid for usual Levi-Civita connections by the theorem (\ref{tc1}). Note that this $E$-torsion is symmetric in the components $b$ and $c$, which is impossible for the usual case except for the torsion-free case. 

\begin{Definition} An almost-Leibniz (respectively pre-Leibniz) algebroid $(E, \rho, [\cdot,\cdot]_E)$ is called an almost-Lie (respectively pre-Lie) algebroid if $[\cdot,\cdot]_E$ is anti-symmetric \cite{16}. Moreover, if the Leibniz identity (\ref{ec20}) holds, it coincides with the Jacobi identity, and $(E, \rho, [\cdot,\cdot]_E)$ is called a Lie algebroid. 
\label{dc15}
\end{Definition}

\begin{Proposition} Any almost-Lie algebroid $(E, \rho, [\cdot,\cdot]_E)$ becomes a local almost-Leibniz algebroid with any locality operator $L \in [0]$.
\end{Proposition}
\begin{proof} The bracket is anti-symmetric, so
\begin{align} [f u, v]_E &= - [v, f u]_E \nonumber\\
&= - \left\{ \rho(v)(f) u + f [v, u]_E \right\} \nonumber\\
&= - \rho(v)(f) u + f [u, v]_E + 0. \nonumber
\end{align}
\label{pc10}
\end{proof}

\noindent One should note that $L = 0$ does not mean that the bracket is anti-symmetric. Moreover, in general there is no reason to expect that 0 is the only element in $[0]$, as there would not be a local $E$-coframe of the type $(Dx^i)$ for some $\{ x^i \} \subset C^{\infty}(M, \mathbb{R})$. Nevertheless, recall that this was the case for the tangent bundle.

\begin{Corollary} $E$-Levi-Civita connections exist for almost-Lie algebroids endowed with a null locality operator.
\label{cc5}
\end{Corollary}
\begin{proof} Almost-Lie brackets are anti-symmetric, so by the proposition (\ref{pc4}), their set defines an affine space $\mathfrak{B}_A$, which is modeled on the anti-symmetric $(1, 2)$-type $E$-tensors. For the locality structure $[0]$, $^0 \mathfrak{T}: \mathfrak{C} \to \ ^{[0]} \mathfrak{B}_A$ is an affine surjection by the $+$ case of the proposition (\ref{pc7}) as one should consider $(1, 2)$-type anti-symmetric $E$-tensor $Z$ in the equation (\ref{ec33}) so that both sides vanish. Hence, $E$-torsion-free connections exist.
\end{proof}

\noindent Due to the anti-symmetry of the bracket, for an almost-Lie algebroid the anholonomy coefficients satisfy 
\begin{equation} \gamma^a_{\ b c} = - \gamma^a_{\ b c},
\label{ec41}
\end{equation}
for every local $E$-frame $(X_a)$.

\begin{Definition} A local almost-Leibniz algebroid whose $E$-anholonomy coefficients for every local $E$-frame satisfy (\ref{ec41}) will be referred as of Lie-type. Moreover, a local $E$-frame that satisfies (\ref{ec41}) will also be referred as a local $E$-frame of Lie-type.
\label{dc16}
\end{Definition}

\noindent The anti-symmetry of the anholonomy coefficients might be useful. For example, on a Lie-type local almost-Leibniz algebroid, $E$-torsion components (\ref{ec39}) of an $E$-Koszul connection simplify as
\begin{equation} ^L T(^K \nabla)^a_{\ b c} = \frac{1}{2} \left[ \Gamma(^K \nabla)^e_{\ d c} L^{a d}_{\ \ e b} + \Gamma(^K \nabla)^e_{\ d b} L^{a d}_{\ \ e c} \right].
\label{ec42}
\end{equation}
Hence the condition for an $E$-Koszul connection to be an $E$-Levi-Civita connection becomes $\Gamma(^K \nabla)^e_{\ d c} L^{a d}_{\ \ e b} = - \Gamma(^K \nabla)^e_{\ d b} L^{a d}_{\ \ e c}$. Note that this is similar to the general condition for $E$-metric-$g$-compatibility on a $g$-orthonormal frame, namely $\Gamma(\nabla)^d_{\ a b} g_{d c} = - \Gamma(\nabla)^d_{\ a c} g_{d b}$ by the equation (\ref{ec7}).

Trivially, all almost-Lie algebroids are of Lie-type, but there are Lie-type local almost-Leibniz algebroids which are not almost-Lie algebroids. Yet, by some additional assumptions, being of Lie-type can be made into a necessary and sufficient condition for being an almost-Lie algebroid.

\begin{Proposition} Let $(E, \rho, [\cdot,\cdot]_E, L)$ be a local almost-Leibniz algebroid such that $L$ satisfies the following property on a local $E$-frame $(X_a)$
\begin{equation} u^a L(D v^b, X_b, X_a) + v^b L(D u^a, X_a, X_b) = 0,
\label{ec43}
\end{equation}
for all $u^a, v^b \in C^{\infty}(M, \mathbb{R})$. Then $(E, \rho, [\cdot,\cdot]_E, L)$ is of Lie-type if and only if $(E, \rho, [\cdot,\cdot]_E)$ is an almost Lie algebroid.
\label{pc11}
\end{Proposition}
\begin{proof} Let $(E, \rho, [\cdot,\cdot]_E, L)$ is an almost Lie algebroid, then by anti-symmetry, for any local $E$-frame $(X_a)$
\begin{equation} [X_a, X_b]_E = \gamma^c_{\ a b} X_c = - [X_b, X_a] = - \gamma^c_{\ ba} X_c, \nonumber
\end{equation}
and by the linear independence $\gamma^c_{\ a b} = - \gamma^c_{\ b a}$.  Hence, without assuming the equation (\ref{ec43}), $(E, \rho, [\cdot,\cdot]_E, L)$ is of Lie-type. Now assume that $(E, \rho, [\cdot,\cdot]_E, L)$ is of Lie-type such that the equation (\ref{ec43}) is satisfied. Let $u, v \in \mathfrak{X}(E)$, and on a local $E$-frame $(X_a)$, let $u = u^a X_a, v = v^b X_b$ for some $u^a, v^b \in C^{\infty}(M, \mathbb{R})$
\begin{align} 
[u, v]_E &= [u^a X_a, v^b X_b]_E \nonumber\\
&= - v^b \rho(X_b)(u^a) X_a + u^a \rho(X_a)(v^b) X_b + u^a v^b [X_a, X_b] + v^b L(D u^a, X_a, X_b) \nonumber\\
&= - v^b \rho(X_b)(u^a) X_a + u^a \rho(X_a)(v^b) X_b - u^a v^b [X_b, X_a] - u^a L(D v^b, X_b, X_a) \nonumber\\
&= -[v, u]_E, \nonumber
\end{align}
where while reversing the signs, both assumptions are used. Hence the bracket is anti-symmetric so that $(E, \rho, [\cdot,\cdot]_E, L)$ is an almost-Lie algebroid.
\end{proof}

\noindent Note that if $L$ is symmetric in $u$ and $v$, then the condition (\ref{ec43}) is equivalent to that $L \in [0]$ as $D$ is a derivation on smooth functions, and $L$ is $C^{\infty}(M, \mathbb{R})$-multilinear.

\begin{Corollary} Given $L \in [0], (E, \rho, [\cdot,\cdot]_E, L)$ is of Lie-type if and only if $(E, \rho, [\cdot,\cdot]_E)$ is an almost Lie algebroid.
\label{cc6}
\end{Corollary}
\begin{proof} For any $L \in [0], L(Df, u, v) = 0$ for all $f \in C^{\infty}(M, \mathbb{R}), u, v \in \mathfrak{X}(E)$. Hence, in particular both terms in the right-hand side of the equation (\ref{ec43}) vanish, and the equation is satisfied.
\end{proof}

\noindent This little corollary is the reason behind the fact that one would not have to distinguish between the anti-symmetry of the bracket and anti-symmetry of the anholonomy coefficients when dealing with Lie algebroids, and in particular the tangent bundle.

\begin{Corollary} Let $(E, \rho, [\cdot,\cdot]_E, 0)$ be a local almost-Leibniz algebroid which is not of Lie-type. Then, there is no $E$-Levi-Civita connections, and there is no holonomic local $E$-frames.
\label{cc7}
\end{Corollary}
\begin{proof} Let $\nabla$ be any $E$-Levi-Civita connection. For the case $L = 0$, the equation (\ref{ec31}) with $\nabla' = \ ^K \nabla$ yields
\begin{equation} \Delta(\nabla, \ ^K \nabla)^a_{\ b c} - \Delta(\nabla, \ ^K \nabla)^a_{\ c b} = - \frac{1}{2} ( \gamma^a_{\ b c} + \gamma^a_{\ c b}), \nonumber
\end{equation}
by the equation (\ref{ec39}). Note that the left-hand side is anti-symmetric in $b$ and $c$ while the right-hand side is symmetric. This forces both sides to vanish. On the other hand, the bracket is not anti-symmetric by the corollary (\ref{cc6}), so the $E$-anholonomy coefficients are not anti-symmetric in lower indices. This creates a contradiction which implies that there is no $E$-Levi-Civita connection. Existence of a holonomic local $E$-frame imply that $E$-torsion and $E$-non-metricity tensors of any $E$-Koszul connection vanish. Yet this is not the case for any non-holonomic frame, which is again contradictory. Therefore, there cannot exist any holonomic frame if $E$ is not of Lie-type.
\end{proof}

\begin{Theorem} On a local almost-Leibniz algebroid equipped with a null locality operator $(E, \rho, [\cdot,\cdot]_E, 0)$, there exists a unique $E$-Koszul connection $\ ^K \nabla$ defined by (\ref{dc14}). Moreover the components of any $E$-connection $\nabla$ can be decomposed as
\begin{align} \Gamma(\nabla)^a_{\ b c} = & \ \Gamma(^K \nabla)^a_{\ b c} + \frac{1}{2} g^{a d} \Big[- Q(\nabla, g)_{b d c} + Q(\nabla, g)_{d c b} - Q(\nabla, g)_{c b d} \nonumber\\
& \qquad \qquad \qquad \qquad - g_{e c} \ ^0 T(\nabla)^e_{\ b d} + g_{e d} \ ^0 T(\nabla)^e_{\ b c} - g_{e b} \ ^0 T(\nabla)^e_{\ c d} \Big].
\label{ec44} 
\end{align}
\label{tc2}
\end{Theorem}
\begin{proof} For $L = 0$, the components of an $E$-Koszul connection (\ref{ec37}) yield
\begin{align} \Gamma(^K \nabla)^a_{\ b c} = & \frac{1}{2} g^{a d} \Bigl[ \rho(X_b) \left( g_{c d} \right) + \rho(X_c) \left( g_{b d} \right) - \rho(X_d) \left( g_{b c} \right) \nonumber\\
& \qquad \quad - \gamma^e_{\ c d} g_{e b} - \gamma^e_{\ b d} g_{e c} + \gamma^e_{\ b c} g_{e d} \Bigr]. \nonumber
\end{align}
As the $\Gamma(^K \nabla)$ dependence of the right-hand side vanishes, the solution for the $E$-connection coefficients is automatically unique. Moreover, one can use Schouten's trick \cite{12} exactly in the same way to find the symmetric and anti-symmetric parts (in lower indices) of the $E$-connection components. As $L = 0$, by the equation (\ref{ec18}), one has for the anti-symmetric part (times two)
\begin{equation} \Gamma(\nabla)^a_{\ b c} - \Gamma(\nabla)^a_{\ c b} = \ ^0 T(\nabla)^a_{\ b c} + \gamma^a_{\ b c}, \nonumber
\end{equation}
which is also true for the $E$-Koszul connection coefficients $\Gamma(^K \nabla)$ in particular. For the symmetric part, consider the following combination on a local $E$-frame $(X_a)$
\begin{align}
Q(^K \nabla, g)_{a b c} - Q(^K \nabla, g)_{b c a} + Q(^K \nabla, g)_{c a b} &= \rho(X_a)(g_{b c}) - \rho(X_b)(g_{c a}) + \rho(X_c)(g_{a b}) \nonumber\\
& \quad - g_{c d} \left[ \Gamma(^K \nabla)^d_{\ a b} - \Gamma(^K \nabla)^d_{\ a b} \right] \nonumber\\
& \quad  - g_{b d} \left[ \Gamma(^K \nabla)^d_{\ a c} + \Gamma(^K \nabla)^d_{\ c a} \right] \nonumber\\
& \quad - g_{a d} \left[ \Gamma(^K \nabla)^d_{\ c b} - \Gamma(^K \nabla)^d_{\ b c} \right] \nonumber\\
& = \rho(X_a)(g_{b c}) - \rho(X_b)(g_{c a}) + \rho(X_c)(g_{a b}) \nonumber\\
& \quad - g_{c d} \left[ ^0 T(^K \nabla)^d_{\ a b} - \gamma^d_{\ a b} \right] \nonumber\\
& \quad  - g_{b d} \left[ \Gamma(^K \nabla)^d_{\ a c} + \Gamma(^K \nabla)^d_{\ c a} \right] \nonumber\\
& \quad - g_{a d} \left[ ^0 T(^K \nabla)^d_{\ c b} - \gamma^d_{\ c b} \right]. \nonumber
\end{align}
After singling out $\Gamma(^K \nabla)^e_{\ a c} + \Gamma(^K \nabla)^e_{\ c a}$ and multiplying both sides by $g^{e b}$, this yields the symmetric part (times two)
\begin{align}
\Gamma(^K \nabla)^e_{\ a c} + \Gamma(^K \nabla)^e_{\ c a} &= g^{e b} \Big\{ - Q(^K \nabla, g)_{a b c} + Q(^K \nabla, g)_{b c a} - Q(^K \nabla, g)_{c a b} \nonumber\\
& \quad + \rho(X_a)(g_{b c}) - \rho(X_b)(g_{c a}) + \rho(X_c)(g_{a b}) \nonumber\\
& \quad - g_{c d} \left[ ^0 T(^K \nabla)^d_{\ a b} - \gamma^d_{\ a b} \right] - g_{a d} \left[ ^0 T(^K \nabla)^d_{\ c b} - \gamma^d_{\ b c} \right]. \Big\} \nonumber
\end{align}
Combining the symmetric and anti-symmetric parts, and using $g^{a d} g_{e d} = \delta^a_{\ e}$, one can get the desired result.
\end{proof}

\begin{Corollary} On an almost-Lie algebroid $(E, \rho, [\cdot,\cdot]_E, 0)$ equipped with an $E$-metric, there exists a unique $E$-Levi-Civita connection. 
\label{cc8}
\end{Corollary}
\begin{proof}
If $E$ is an almost-Lie algebroid, then it is of Lie-type by the corollary (\ref{cc6}). Hence as $L = 0$, the $E$-torsion (\ref{ec40}) and $E$-non-metricity (\ref{ec41}) components of the unique $E$-Koszul connection $^K \nabla$ vanish so that $^K \nabla$ coincides with the unique $E$-Levi-Civita connection.
\end{proof}
\noindent This corollary trivially implies the fundamental theorem of Riemannian geometry, as the tangent bundle is a Lie algebroid.

An anchored vector bundle $(E, \rho)$ is said to be transitive if the anchor $\rho$ is surjective. A surjective vector bundle morphism is automatically of locally constant rank, so that $(E, \rho)$ is regular. Hence, one trivially has the following short exact sequence of vector bundles
\begin{equation} 0 \to ker(\rho) \xrightarrow{i} E \xrightarrow{\rho} T(M) \to 0,
\label{ec45}
\end{equation}
where $i: ker(\rho) \to E$ is the inclusion map. If $M$ is paracompact, then like any short exact sequence of vector bundles over $M$, this exact sequence splits in the category of vector bundles \cite{18}. Hence, there exists a vector bundle isomorphism $\tau: E \to T(M) \oplus ker(\rho)$ satisfying
\begin{equation} \tau \circ g^{-1} \circ \rho^* = \left( 0, id_{ker(\rho)} \right), \qquad \rho = proj_1 \circ \tau,
\label{ec46}
\end{equation}
where $proj_i$ denotes the projection onto $i$th component in a direct sum. Note that if $(E, \rho, [\cdot,\cdot]_E)$ is a pre-Leibniz algebroid, then $(ker(\rho), \rho|_{ker(\rho)} = 0, [\cdot,\cdot]_E|_{ker(\rho)})$ is a pre-Leibniz algebroid. Hence, in this case (\ref{ec39}) becomes an exact sequence of pre-Leibniz algebroids. Detailed work on this topic can be found in \cite{17}. Note that $proj_2 \circ \tau$ is a locality projector, trivially. Transitive pre-Leibniz algebroids are of importance due to the fact that the usual vector fields generate diffeomorphisms, and one wants to have all the diffeomorphisms in a consistent physical theory.

%%%%%%%%%%%%%%%%%%%%%%%%%%%%%%%%%%%%%%%%%%%%%%%%%%%%%%

\section{Generalized Geometry on Exact Courant Algebroids}

\noindent The aim of this section is to construct generalized geometry as a special case of $E$-metric-connection geometries on local pre-Leibniz algebroids. A suitable ``generalization'' for the local double field theory would be extended Riemannian geometry, which uses the language of symplectic pre-$NQ$-manifolds and $L_{\infty}$-structures \cite{19}. Here, generalized geometry will be constructed by using the exact Courant algebroids, which corresponds to the special case for extended geometry when one considers symplectic Lie 2-algebroids \cite{20}. We will define necessary structures, that already exist in the literature, in a parallel way to $E$-metric-connection geometries. Similar to the previous section, we will construct exact Courant algebroids by adding assumptions step by step. Hence, we first consider almost-metric algebroids. 

\begin{Definition} An almost-metric algebroid over $M$ is a quadruplet $(E, \rho, [\cdot,\cdot]_E, g)$, where $(E, \rho)$ is an anchored vector bundle over $M$, $[\cdot,\cdot]_E$ is a bracket on $E$, and $g$ is an $E$-metric such that
\begin{equation} [u, v]_E + [v, u]_E = \mathcal{D}_g (g(u, v)),
\label{ed1}
\end{equation}
for all $u, v \in \mathfrak{X}(E)$, where $\mathcal{D}_g := g^{-1} \circ D$.
\label{dd1}
\end{Definition}

\begin{Definition} An almost-metric algebroid $(E, \rho, [\cdot,\cdot]_E, g)$ is called an almost-Courant algebroid if $(E, \rho, [\cdot,\cdot]_E)$ is an almost-Leibniz algebroid over $M$. 
\label{dd2}
\end{Definition}

\noindent One way to construct an almost-Courant algebroid is to assume a compatibility condition between all the structures of an almost-metric algebroid $(E, \rho, [\cdot,\cdot]_E, g)$ in the sense
\begin{equation} \rho(u)(g(v, w)) = g([u, v]_E, w) + g(u, [v, w]_E),
\label{ed2}
\end{equation}
for all $u, v, w \in \mathfrak{X}(E)$. If this is the case, then an almost-metric algebroid is called a metric-algebroid \cite{21}. 

\begin{Proposition} Any almost-Courant algebroid $(E, \rho, [\cdot,\cdot]_E, g)$ can be seen as a local almost-Leibniz algebroid $(E, \rho, [\cdot,\cdot]_E, \ ^g L)$, where the locality operator $^g L$ is defined by
\begin{equation}
^g L(\Omega, u, v) := g(u, v) g^{-1}(\Omega),
\label{ed3}
\end{equation}
for all $\Omega \in \Omega^1, u, v \in \mathfrak{X}(E)$.
\label{pd1}
\end{Proposition}
\begin{proof} By the equation (\ref{ed1})
\begin{align}
[f u, v]_E &= - [v, f u]_E + \mathcal{D}_g(g(f u, v)), \nonumber\\
&= - \left\{ \rho(v)(f) u + f [v, u]_E \right\} + \mathcal{D}_g(f g(u, v)), \nonumber\\
&= - \rho(v)(f) u - f \left\{ - [u, v]_E + \mathcal{D}_g(g(u, v)) \right\} + \mathcal{D}_g(f) g(u, v) + f \mathcal{D}_g(g(u, v)) , \nonumber\\
&= - \rho(v)(f) + f [u, v]_E + \mathcal{D}_g(f) g(u, v), \nonumber
\end{align}
for all $f \in C^{\infty}(M, \mathbb{R}), u, v \in \mathfrak{X}(E)$. For the locality operator in the equation (\ref{ed3}), insert $\Omega = Df$ for some $f \in C^{\infty}(M, \mathbb{R})$
\begin{equation}
^g L(Df, u, v) := g(u, v) g^{-1}(Df) = g(u, v) \mathcal{D}_g(f). \nonumber
\end{equation}
Hence $^g L$ or any other locality operator from the same locality structure satisfy for the right-Leibniz identity.
\end{proof}

\noindent Note that $^g L(\Omega, u, v)$ is symmetric in $u$ and $v$ due to the $E$-metric $g$. The components of the locality $E$-tensor (\ref{ed3}) over a local $E$-frame $(X_a)$ read
\begin{equation} ^g L^{a b}_{\ \ c d} = g^{a b} g_{c d},
\label{ed5}
\end{equation}
where $g^{a b} := g^{-1}(e^a, e^b)$. For a $g$-orthonormal local $E$-frame, the components of this locality tensor are widely used in the double field theory literature.

\begin{Definition} An almost-Courant algebroid $(E, \rho, [\cdot,\cdot]_E, g)$ is called a pre-Courant algebroid if $\rho([u, v]_E) = [\rho(u), \rho(v)]$ for all $u, v \in \mathfrak{X}(E)$\footnote{For a slightly more restrictive definition of pre-Courant algebroids, see \cite{22}.}.
\label{dd3}
\end{Definition}

\begin{Definition} A metric algebroid $(E, \rho, [\cdot,\cdot]_E, g)$ is called a Courant algebroid if $(\mathfrak{X}(E), [\cdot,\cdot]_E)$ is a Leibniz algebra.
\label{dd4}
\end{Definition}

\noindent As metric-algebroids are almost-Courant algebroids, in particular Courant algebroids are almost-Courant algebroids. Moreover, due to the fact that defining property Leibniz identity (\ref{ec20}) implies that the anchor preserves the brackets, all Courant algebroids are pre-Courant algebroids.
 
%\begin{Proposition} For an almost-Courant algebroid $(E, \rho, [\cdot,\cdot]_E, g)$, one can show that $\rho \circ \mathcal{D}_g = 0$ so that $\rho \circ g^{-1} \circ \rho^* = 0$, and $im(g^{-1} \circ \rho^*) \subset ker(\rho)$. 
%\label{pd2}
%\end{Proposition}
%\begin{proof}
%\end{proof}

\begin{Definition} Let $(E, \rho, [\cdot,\cdot]_E, g)$ be an almost-Courant algebroid. If the following sequence is a short exact sequence of vector bundles
\begin{equation} 0 \rightarrow T^*(M) \xrightarrow{g^{-1} \circ \rho^*} E \xrightarrow{\rho} T(M) \rightarrow 0,
\label{ed6}
\end{equation}
then the almost-Courant algebroid is said to be exact. 
\label{dd5}
\end{Definition}

\noindent In this case, $g^{-1} \circ \rho^*$ is injective, and its image is isomorphic to $ker(\rho)$, so that $ker(\rho) \simeq T^*(M)$. Hence, this exact sequence is of the form (\ref{ec45}), and similarly when $M$ is paracompact, one can deduce that there is a vector bundle isomorphism $\tau: E \to T(M) \oplus T^*(M)$ with $\tau \circ g^{-1} \circ \rho^* = \left( 0, id_{T^*(M)} \right), \rho = proj_1 \circ \tau$. Hence, any exact almost-Courant algebroid is a rank $2n$ vector bundle over an $n$ dimensional manifold.

The ``untwisted generalized-tangent bundle'' $\mathbb{T}(M) := T(M) \oplus T^*(M)$ is an exact Courant algebroid, with the anchor being the projection map $proj_1: \mathbb{T}(M) \to T(M)$, the Dorfman bracket\footnote{Its anti-symmetrization is called the Courant bracket.} $[\cdot,\cdot]_D: \Gamma(\mathbb{T}(M)) \times \Gamma(\mathbb{T}(M)) \to \Gamma(\mathbb{T}(M))$
\begin{equation} [U + \omega, V + \eta]_D := [U, V] + \mathcal{L}_U \eta - \iota_V d \omega,
\label{ed7}
\end{equation}
and the $\mathbb{T}(M)$-metric $(\cdot,\cdot): \Gamma(\mathbb{T}(M)) \times \Gamma(\mathbb{T}(M)) \to C^{\infty}(M, \mathbb{R})$ defined by
\begin{equation} (U + \omega, V + \eta) := \frac{1}{2} \left( \iota_U \eta + \iota_V \omega \right),
\label{ed8}
\end{equation}
where $U, V \in \mathfrak{X}(M), \omega, \eta \in \Omega^1(M)$. Note that this $\mathbb{T}(M)$-metric is $(n, n)$ signature. Hence, the group of unitary transformations is the orthogonal group $O(n, n)$, which is suitable for $T$-duality.

Any exact Courant algebroid can be constructed by starting from the untwisted generalized-tangent bundle $\mathbb{T}(M)$ and $H$-twisting it via the bracket $[\cdot,\cdot]_H: \Gamma(\mathbb{T}(M)) \times \Gamma(\mathbb{T}(M)) \to \Gamma(\mathbb{T}(M))$
\begin{equation} [V + \omega, W + \eta]_H := [V + \omega, W + \eta]_D + \iota_W \iota_V H.
\label{ed9}
\end{equation}
for a closed 3-form $H \in \Omega^3_{cl}(M)$. The cohomology class $[H]$, which is called the \v{S}evera class, completely classifies exact Courant algebroids up to Courant algebroid isomorphism \cite{23}. 

One can define the ``generalized'' versions of the geometric structures in the sense of Hitchin. Let $(E, \rho, [\cdot,\cdot]_E, g)$ be an almost-Courant algebroid. Then a $(q, r)$-type generalized-tensor is just a $(q, r)$-type $E$-tensor, and everything about $E$-tensors, $E$-vector fields etc. can be carried on here as generalized versions.

\begin{Definition} On an almost-Courant algebroid $(E, \rho, [\cdot,\cdot]_E, g)$, a generalized-metric $G$ is an $E$-metric which satisfies
\begin{equation} g(u, v) = g^{-1}(G(u), G(v)),
\label{ed10}
\end{equation}
for all $u, v \in \Gamma(E)$.
\label{dd6}
\end{Definition}

\noindent The condition (\ref{ed10}) can be written as 
\begin{equation} g_{a b} = g^{c d} G_{a c} G_{b d}.
\label{ed11}
\end{equation}
on a local generalized-frame $(X_a)$, which is also widely used in double field theory literature. Any generalized-metric on an exact Courant algebroid splits $E$ into $E = C_+ \oplus C_-$ where $(n, n)$ signature $g$ reduces to $(t, s)$ and $(s, t)$ signatures on $C_+$ and $C_-$ respectively for some $t, s \in \mathbb{N}$ with $t + s = n$, where $C_-$ is the orthogonal complement of $C_+$. Note that, $g$ itself becomes a generalized-metric if components of $g$ and $g^{-1}$ coincide, i. e. $g^2 = id_{\mathfrak{X}(M)}$\footnote{This fact has relations with the para-Hermitian formulation of double field theory \cite{24}.}. This is the case for (\ref{ed8}), so it holds for any exact Courant algebroid. Any generalized-metric on the untwisted generalized tangent bundle can be written as
\begin{equation} G = \left( \begin{matrix} \hat{G} - B \hat{G}^{-1} B & B \hat{G}^{-1} \\ -\hat{G}^{-1} B & \hat{G}^{-1} \end{matrix} \right),
\label{ed12}
\end{equation}
where $B \in \Omega^2(M)$, and $\hat{G}$ is a usual $(t, s)$ signature metric over $M$. Moreover any $G$ of this form is a generalized-metric. Twisting corresponds to $B$ being a locally 2-form, but globally a connective structure on a bundle gerbe.

\begin{Definition} On an almost-Courant algebroid $(E, \rho, [\cdot,\cdot]_E, g)$, a linear $E$-connection $\nabla$ on $(E, \rho)$ is called a linear generalized-connection. The generalized-non-metricity tensor corresponding to a linear generalized-connection $\nabla$ and a generalized-metric $G$ is the $(0, 3)$-type $E$-non-metricity tensor of $\nabla$ and $G$. The generalized-torsion operator corresponding to $\nabla$ is the $E$-torsion operator of $\nabla$ on the local almost-Leibniz algebroid $(E, \rho, [\cdot,\cdot]_E, \ ^g L)$.
\label{dd7}
\end{Definition}
\noindent This generalized-torsion operator explicitly reads\footnote{As in the previous section, any $L \in [^g L]$ would work.}
\begin{equation} ^{^g L} T(\nabla)(u, v) = \nabla_u v - \nabla_v u - [u, v]_E + g(\nabla_{X_a} u, v) g^{-1}(e^a),
\label{ed13}
\end{equation}
for all $u, v \in \Gamma(E)$, where $(X_a)$ is a local generalized-frame, with its dual local generalized-coframe $(e^a)$\footnote{i. e. $(X_a)$ is a local $E$-frame on an almost-Courant algebroid $E$, and $(e^a)$ is its dual.}. 

\begin{Definition} On a regular pre-Courant algebroid $(E, \rho, [\cdot,\cdot]_E, g)$, the generalized-curvature operator corresponding to a generalized-connection $\nabla$ is defined as the $E$-curvature operator of $\nabla$ on the local pre-Leibniz algebroid $(E, \rho, [\cdot,\cdot]_E, \ ^g L)$ endowed with a locality projector $\mathcal{P}$ on $[^g L]$
\label{dd8}
\end{Definition}

\noindent This generalized-curvature operator explicitly reads
\begin{equation}
^{\widehat{^g L}} R(\nabla)(u, v, w) = \nabla_u \nabla_v w - \nabla_v \nabla_u w - \nabla_{[u, v]_E} w + \nabla_{{\widehat{^g L}} (e^a, \nabla_{X_a} u, v))} w,
\label{ed14} 
\end{equation}
for all $u, v, w \in \Gamma(E)$, where $(X_a)$ is a local generalized-frame.

Exact almost-Courant algebroids, so in particular exact Courant algebroids, are regular by definition because the anchor is surjective. Moreover, Courant algebroids are pre-Courant algebroids, so generalized-curvature operator can be defined on any exact Courant algebroid without any additional assumption.

\begin{Proposition} Let $(E, \rho, [\cdot,\cdot]_E, \ g)$ be an exact almost-Courant algebroid, then there is a unique locality projector $\mathcal{P}$ given by the projection onto $ker(\rho) \simeq T^*(M)$ given by $proj_2 \circ \tau$, where $\tau: E \to T(M) \oplus T^*(M)$ is the isomorphism coming from the exact sequence (\ref{ed5}).
\label{pd2}
\end{Proposition} 
\begin{proof} For notational ease, we will ignore the isomorphism $\tau$ and directly take $E = T(M) \oplus T^*(M)$. As the kernel of the anchor $\rho$ is $T^*(M)$, a locality projector is of the form $(U, \omega) \mapsto (0, \tilde{\omega})$ for $U \in \mathfrak{X}(M), \omega, \tilde{\omega} \in \Omega^1(M)$. If $(U, \omega)$ is in the image of the coboundary map $D$, which coincides with the exterior derivative $d$, then its image under the locality projector will be the same as itself. This implies $(0, \omega) \mapsto (0, \omega)$ for exact $\omega \in \Omega^1(M)$, as $g^{-1} \circ \rho^*$ is just the inclusion. Any 1-form can be written as a sum of exact 1-forms locally. Hence, this gives $(0, \omega) \mapsto (0, \omega)$ for any $\omega \in \Omega^1(M)$. Therefore, there is only one possible locality projector, which is given by the projection map onto the cotangent bundle.
\end{proof}

\noindent This justifies the choice of generalized-curvature operators in the double field theory literature defined with the projection onto $T^*(M)$. 
%Moreover, this might be related to ``stringy differential geometry'' with a projection [24 - Park].

In the double field theory literature, generalized-torsion and generalized-curvature tensors are usually defined index-down by using the isomorphism $g^{-1}: \Omega^1(E) \to \mathfrak{X}(E)$. Here, as in the previous section, they will be defined completely analogous to the usual case
\begin{align} ^{^g L} T(\nabla)(\Omega, u, v) &:= \langle \Omega, \ ^{^g L} T(\nabla)(u, v) \rangle, \nonumber\\
^{\widehat{^g L}} R(\nabla)(\Omega, u, v, w) &:= \langle \Omega, \ ^{\widehat{^g L}} R(\nabla)(u, v, w) \rangle,
\label{ed15}
\end{align}
for all $\Omega \in \Omega^1(E), u, v, w \in \mathfrak{X}(E)$. For exact Courant algebroids, another way to have $E$-tensorial curvature would be to define a generalized-curvature with respect to the pseudo-curvature map $R^{(0)}(\nabla)$ and then restrict it to a Dirac structure so that it is $C^{\infty}(M, \mathbb{R})$-multilinear \cite{25}. Generalized-Ricci tensor and generalized-Ricci scalar can be defined as in the pre-Leibniz algebroid case. 

\noindent In this setting, we have the necessary structures to define generalized-metric-connection geometries.

\begin{Definition} Generalized-metric-connection geometries are defined as a quadruplet $(M, (E, \rho, [\cdot,\cdot]_E, g, \mathcal{P}), G, \nabla)$ where $M$ is a manifold, $(E, \rho, [\cdot,\cdot]_E, g)$ is an exact Courant algebroid over $M$, $\mathcal{P} = proj_2$ is the unique locality projector in the locality structure $[^g L]$, $G$ is a generalized-metric, and $\nabla$ is a linear generalized-connection. 
\label{dd9}
\end{Definition}

\noindent Note that these geometries can be defined on any regular pre-Courant algebroid. By construction, this definition is equivalent to the $E$-metric-connection geometry $(M, (E, \rho, [\cdot,\cdot]_E, \ ^g L, \mathcal{P}), G, \nabla)$ when $G$ satisfies (\ref{ed10}). Hence, in this sense $E$-metric-connection geometries ``generalize'' generalized geometries.

\begin{Definition} If a linear generalized-connection is generalized-metric-compatible with both $g$ and $G$ and generalized-torsion-free , then it is called a generalized-Levi-Civita connection associated to $G$.
\label{dd10}
\end{Definition}

\noindent Even though there are two generalized-metric-compatibility conditions, generalized-Levi-Civita connections are still not unique, except in some small number of special cases \cite{26}. 

\begin{Proposition} On a local almost-Leibniz algebroid $(E, \rho, [\cdot,\cdot]_E, \ ^g L)$ endowed with an $E$-metric $g$, assume that $\nabla$ is an $E$-Levi-Civita-connection. On a $g$-orthonormal local $E$-frame, if $E$-connection coefficients $\{ \Gamma(\nabla)^a_{\ b c} \}$ are symmetric (respectively anti-symmetric) in $b$ and $c$, then the $E$-anholonomy coefficients $\{ \gamma^a_{\ b c} \}$ have to be also symmetric (respectively anti-symmetric) in $b$ and $c$.
\label{pd3}
\end{Proposition}
\begin{proof} On a $g$-orthonormal local $E$-frame, $E$-non-metricity components (\ref{ec7}) for a linear $E$-connection $\nabla$ become
\begin{equation} Q(\nabla, g)_{a b c} = - \Gamma(\nabla)^d_{\ a b} g_{d c} - \Gamma(\nabla)^d_{\ a c} g_{b d}, \nonumber
\end{equation}
on a local $E$-frame $(X_a)$. Assuming metric-$g$-compatibility and multiplying by $g^{e c}$, this becomes
\begin{equation} \Gamma(\nabla)^e_{\ a b} = - \Gamma(\nabla)^d_{\ a c} g^{e c} g_{d b},
\label{ed16}
\end{equation}
as $g$ is symmetric. By the equation (\ref{ed13}), $E$-torsion-free condition on components can be written as
\begin{equation} \Gamma(\nabla)^a_{\ b c} - \Gamma(\nabla)^a_{\ c b} - \gamma^a_{\ b c} + \Gamma(\nabla)^e_{\ d b} g^{a d} g_{e c} = 0. \label{ed17}
\end{equation}
Note that the last term in the summation is the same as the right-hand side of the equation (\ref{ed16}) if one changes the indices $d$ and $b$ in $\Gamma(\nabla)^e_{\ d b}$. If $\Gamma(\nabla)^e_{\ d b}$ is symmetric or anti-symmetric in $d$ and $b$, one can use this information. First, let us assume that it is symmetric, then the equation (\ref{ed17}) becomes
\begin{equation} \Gamma(\nabla)^a_{\ b c} - \Gamma(\nabla)^a_{\ b c} - \gamma^a_{\ b c} - \Gamma(\nabla)^a_{\ b c} = 0, \nonumber
\end{equation}
by the equation (\ref{ed16}). This yields 
\begin{equation} \gamma^a_{\ b c} = - \Gamma(\nabla)^a_{\ b c}, \nonumber
\end{equation}
forcing $\gamma^a_{\ b c}$ to be also symmetric in $b$ and $c$. Similarly, if one assumes anti-symmetry of $\Gamma(\nabla)^e_{\ d b}$ in $d$ and $b$, then one gets
\begin{equation} \gamma^a_{\ b c} = 3 \Gamma(\nabla)^a_{\ b c}, \nonumber
\end{equation}
which forces $\gamma^a_{\ b c}$ to be anti-symmetric in $b$ and $c$.
\end{proof}

\begin{Proposition} On an almost-metric algebroid $(E, \rho, [\cdot,\cdot]_E, g)$, $g$-orthonormal local $E$-frames are of Lie-type.
\label{pd4}
\end{Proposition}
\begin{proof} By the defining property of almost-metric algebroids (\ref{ed1}),
\begin{equation} [X_a, X_b]_E + [X_b, X_a]_E = \mathcal{D}_g(g(X_a, X_b)), \nonumber
\end{equation}
on a local $E$-frame $(X_a)$. If $(X_a)$ is $g$-orthonormal, the right-hand side vanishes because $\mathcal{D}_g$ can be written as $g^{-1} \circ \rho^* \circ d$. Hence, by the linear independence, one gets $\gamma^c_{\ a b} = - \gamma^c_{b a}$.
\end{proof}

\begin{Corollary} Let $(E, \rho, [\cdot,\cdot]_E, g)$ be an almost-Courant algebroid and $\nabla$ an $E$-Levi-Civita connection. Then, non-zero $E$-connection coefficients $\{ \Gamma(\nabla)^a_{\ b c} \}$ cannot be symmetric in $b$ and $c$ on a $g$-orthonormal local $E$-frame.
\label{cd1}
\end{Corollary}
\begin{proof} Almost-Courant algebroids are almost-metric algebroids by definition. Hence, the proposition (\ref{pd4}) is valid for almost-Courant algebroids, and so $g$-orthonormal local $E$-frames are of Lie-type. Moreover, by the proposition (\ref{pd1}), almost-Courant algebroids are local almost-Leibniz algebroids with the locality operator $^g L$. Therefore, the proposition (\ref{pd3}) is also valid for them, which means that $\Gamma(\nabla)^a_{\ b c}$ cannot be symmetric in $b$ and $c$. 
\end{proof}

\begin{Proposition} On an almost-Courant algebroid $(E, \rho, [\cdot,\cdot]_E, g)$, an $E$-Koszul connection is $E$-torsion-free if and only if it is $E$-metric-$g$-compatible.
\label{pd5}
\end{Proposition}
\begin{proof} By the proposition (\ref{pd4}), $g$-orthonormal local $E$-frames are of Lie-type. Hence, by the proposition (\ref{pd1}), the $E$-torsion (\ref{ec39}) and $E$-non-metricity (\ref{ec40}) components of an $E$-Koszul connection become
\begin{align} 
T(^K \nabla)^a_{\ b c} &= \frac{1}{2} \left( \Gamma(^K \nabla)^e_{\ a c} g_{e b} + \Gamma(^K \nabla)^e_{\ a b} g_{e c}\right) g^{a d}, \nonumber\\
Q(^K \nabla, g)_{a b c} &= - \frac{1}{2} \left( \Gamma(^K \nabla)^e_{\ a c} g_{e b} + \Gamma(^K \nabla)^e_{\ a b} g_{e c}\right), \nonumber
\end{align}
on a $g$-orthonormal local $E$-frame $(X_a)$. They clearly satisfy
\begin{equation} T(^K \nabla)^a_{\ b c} = - Q(^K \nabla, g)_{d b c} g^{a d}. \nonumber
\end{equation}
Hence, if it is $E$-metric-$g$-compatible, then it is $E$-torsion-free. By the proposition (\ref{pc9}), the other implication is already proved for a more general case. 
\end{proof}

%%%%%%%%%%%%%%%%%%%%%%%%%%%%%%%%%%%%%%%%%%%%%%%%%%%%%%

\section{Concluding Remarks}

\noindent In this paper, $E$-metric-connection geometries are constructed on regular local pre-Leibniz algebroids with a locality projector. This construction is done with the fewest possible number of assumptions and completely parallel to usual metric-affine geometries on a smooth manifold. As a special case, metric-affine geometry is deduced in a unique way, so one can say that with this new geometry one generalizes the general relativity. Moreover, on exact almost-Courant algebroids another uniqueness result is proven for the locality projector, which explains the necessity of the curvature operator used in the double field theory literature. By combining the existing information, especially on Lie algebroids, some new results are proven. For example, $E$-Koszul connections, which are generalizations of Levi-Civita connections, are defined and shown to be helpful for numerous properties. In particular, a generalization of the fundamental theorem of Riemannian geometry is proven for almost-Leibniz algebroids. Moreover, some special cases for the existence of $E$-Levi-Civita connections are investigated. Most importantly, locality structures and locality projectors on local almost-Leibniz algebroids are defined in order to construct $E$-curvature tensor.

One may think of several other structures in the usual geometrical setting that can be ``lifted'' to pre-Leibniz algebroids. For example, Weyl invariant theories \cite{27} on pre-Leibniz algebroids would be an interesting case. The authors' ongoing project on this topic by defining $E$-versions of conformal, projective and Weyl structures is on its way. Such constructions might lead one to the use of an analogous version of Riemann-Cartan-Weyl geometry to explain, for example, the M-theoretic supergravity 3-form $C$-field \cite{28}. Another possible direction to extend our work would be to define a local double field theory in terms of a scalar field derived just from a generalized non-metricity tensor. This will be the generalization of the symmetric teleparallel gravity \cite{29}. Moreover, there are some results on the generalizations of deformations \cite{30} and $H$-twisting \cite{31} of Lie brackets, which will be the subject of future work.

\bigskip
\noindent \textbf{Acknowledgment}

\noindent The authors are thankful to Ahmet Berkay Ke\c{c}eci and F\i rt\i na K\"{u}\c{c}\"{u}k for fruitful discussions on abelian and non-abelian categories. They are also grateful to Frederick Reece for proofreading the paper.

\end{document}